\documentclass[numbook, envcountsect, envcountsame, envcountreset, runningheads, smallextended]{svjour3}

\def\makeheadbox{\relax}
\usepackage{graphicx} 
\usepackage{latexsym,amsfonts,amssymb,amsmath,bbm,mathtools} 
\usepackage{booktabs,multirow,subfigure,float}
\usepackage[T1]{fontenc}
\usepackage[ruled, linesnumbered]{algorithm2e}
\usepackage{hyperref}
\usepackage{color}

\makeatletter
      \spnewtheorem*{Proof}{}{\itshape*\bfseries*}{\rmfamily}
      \def\Proof@Opargbegintheorem #1#2#3#4{#4\trivlist \item
      {#3#2\@thmcounterend \ }}
\AtBeginEnvironment{Proof}{\let\@Opargbegintheorem\Proof@Opargbegintheorem}
      \makeatother

\usepackage{etoolbox}
\usepackage{color}
\appto\ProcessRunnHead{\markboth{\truncate{220pt}{\authrun}\csuse{@author}}{\truncate{220pt}{\titrun}{\csuse{@title}}}}
\def\truncate#1#2{\ifdim#1<\wd#2\relax \rlap{\fboxsep=1pt\fboxrule1pt\color{red}\fbox{\hbox to#1{\strut\hss}}}\fi}
\smartqed

\usepackage[backend=biber,natbib=true,bibencoding=utf8,style=numeric]{biblatex}
\newcommand{\IP}{\mathbb P}

\newcommand{\IE}{\mathbb E}
\newcommand{\IR}{\mathbb R}

\newcommand{\IS}{\mathbb S}
\newcommand{\IF}{\mathbb F}

\renewcommand\leq{\leqslant}

\renewcommand\geq{\geqslant}

\newcommand{\cA}{{\mathcal A}}
\newcommand{\cB}{{\mathcal B}}

\newcommand{\cF}{{\mathcal F}}

\newcommand{\cP}{{\mathcal P}}

\newcommand{\cM}{{\mathcal M}}
\newcommand{\cX}{{\mathcal X}}
\newcommand{\indic}[1]{\mathbbm{1}_{#1}}
\renewcommand{\epsilon}{\varepsilon}
\renewcommand{\phi}{\varphi}

\DeclareMathOperator*{\argmax}{arg\,max}

\newcommand*\diff{\mathop{}\!\mathrm{d}}

\def\subclassname{{\bfseries Mathematics Subject Classification
}\enspace}
\def\subclass#1{\par\addvspace\medskipamount{\rightskip=0pt plus1cm
\def\and{\ifhmode\unskip\nobreak\fi\ $\cdot$
}\noindent\subclassname\ignorespaces#1\par}}
\def\JELname{{\bfseries JEL Classification}\enspace}
\def\JEL#1{\par\addvspace\medskipamount{\rightskip=0pt plus1cm
\def\and{\ifhmode\unskip\nobreak\fi\ $\cdot$
}\noindent\JELname\ignorespaces#1\par}}

\spnewtheorem{assumption}[theorem]{Assumption}{\bf}{\rm}

\begin{document}

\journalname{Finance and Stochastics}

\title{Calibration of Local Volatility Models with Stochastic Interest Rates using Optimal Transport}
\titlerunning{Local Volatility Calibration with Stochastic Interest Rates}

\author{Benjamin Joseph\thanks{This research has been supported by BNP Paribas Global Markets and the EPSRC Centre for Doctoral Training in Mathematics
of Random Systems: Analysis, Modelling and Simulation (EP/S023925/1).} \and
        Gr\'{e}goire Loeper \and
        Jan Ob\l{}\'{o}j\thanks{For the purpose of Open Access, the author has applied a CC BY public copyright licence to any Author Accepted Manuscript
(AAM) version arising from this submission.}.
}
\authorrunning{B. Joseph, G. Loeper, J. Ob\l{}\'{o}j} 
\institute{B. Joseph \at
              Mathematical Institute and Christ Church, University of Oxford \\
              \email{benjamin.joseph@maths.ox.ac.uk}          
           \and
           G. Loeper \at
              BNP Paribas Global Markets \\
              \email{gregoire.loeper@bnpparibas.com}
            \and
           J. Ob\l{}\'{o}j \at
              Mathematical Institute and St John's College, University of Oxford \\
              \email{jan.obloj@maths.ox.ac.uk}
}

\date{\today}
\maketitle
\begin{abstract}
We develop a non-parametric, semimartingale optimal transport, calibration methodology for local volatility models with stochastic interest rate. The method finds a fully calibrated model which is the closest, in a way that can be defined by a general cost function, to a given reference model. We establish a general duality result which allows to solve the problem by optimising over solutions to a second order fully non-linear Hamilton-Jacobi-Bellman equation. 
Our methodology is analogous to \textcite{guo2022calibration} and \textcite{guo2022joint} but features a novel element of solving for discounted densities, or sub-probability measures.
As an example, we apply the method to a sequential calibration problem, where a Vasicek model is already given for the interest rates and we seek to calibrate a stock price's local volatility model with volatility coefficient depending on time, the underlying and the short rate process, and the two processes driven by possibly correlated Brownian motions. The equity model is calibrated to any  number of European options prices.
\end{abstract}
\subclass{91G20 \and 91G80 \and 60H30}
\JEL{C02 \and C63 \and G13}

\section{Introduction}
Modelling involves inevitable trade-offs: ``All models are wrong, some models are useful'' as  \textcite{box1987empirical} put it. Models need to capture the important aspects of the system they represent but they also need to be tractable, and analytically and/or numerically solvable. In particular, calibration  -- picking model parameters which recover known outputs -- 
is an essential part of any modelling process. It is a key challenge faced by  financial industry practitioners on a daily basis, as their pricing models need to match market prices of liquid instruments before they can be used to price any bespoke or illiquid products. 

In practice, models for a key underlying, such as the S\&P500 index, will need to be calibrated to a large number of options with different maturities and strikes. This may be nigh impossible for a simple parametric model. \textcite{dupire1994pricing} derived a formula to calibrate a local volatility model to an arbitrary number of options, establishing it as a benchmark for equities modelling (see also \textcite{atlan2006localizing} for its stochastic interest rate extension). However, it came with its own shortcomings. Its calibration poses serious numerical challenges, see  \textcite{bain2021calibration}, and requires interpolation of the data as Dupire's formula assumes a continuum of prices across strikes and maturities. 
Moreover, it has been criticised for wrong dynamic behaviour compared to suitable stochastic volatility models, see \textcite{hagan2002managing}. This issue can lead to potentially costly mispricing of exotic products. Naturally, this motivated further research and led, in particular, to the introduction of local stochastic volatility models (LSV) or forward variance stochasic volatility models, see \textcite{bergomi2015stochastic} for details and historical references. While successful in many ways, these models also came with shortcomings, often related to an inability to calibrate to a new class of products. Prominent among these was the challenge to calibrate jointly to SPX and VIX options, see \textcite{guyon2020joint} for the literature review and discussion.

Monte Carlo methods to calibrate with Dupire's formula have been used in \textcite{deelstra2013local, hok2019calibration, ogetbil2022calibrating}. Another approach is to discretise the non-linear Fokker-Planck equation and directly solve it by finite differences as in \textcite{ren2007calibrating} or finite volume methods as in \textcite{wyns2017finite, engelmann2021calibration}. \textcite{dupire1996unified} derives a non-linear McKean SDE for the local volatility, which in fact recovered the mimicking result of \textcite{gyongy1986mimicking} using a financial argument. This has been solved via Monte Carlo methods developed in \textcite{henry2009calibration} for an approximate calibration, and an exact calibration by the particle method of \textcite{guyon2011smile, guyon2012being}. \textcite{cozma2019calibration} used variance reduction techniques to implement the particle method in a four factor model. All of these methods require a continuum of strikes and maturities to construct a surface of implied volatilities, which in practice results in interpolation and extrapolation of market data to build these surfaces. We adopt a different approach, without this requirement, by formulating the exact calibration of European options as discrete constraints within a convex optimisation problem.

Recently, a non-parametric exact calibration method based on optimal transport has been proposed, which aims to address the above challenges. This method uses semimartingale optimal transport as introduced by \textcite{tan2013optimal}, which is the semimartingale version of the celebrated work of \textcite{benamou2000computational} (see also \textcite{huesmann2019benamou}). 
This approach has already been used in several contexts: for local volatility calibration in \textcite{guo2019local}; for local stochastic volatility calibration in \textcite{guo2022calibration}, for joint VIX SPX calibration in \textcite{guo2022joint} and for optimal investment in \textcite{guoNingloeper1}. The most general formulation of the method, pointing to the breadth of its applications, is found in \textcite{guo2021path}. 
 
An overview of these results is given in \textcite{guo2022optimal}. It is worth noting that whilst the link with optimal transport was not recognised, a variational approach to calibrating local volatility models was first constructed already in \textcite{avellaneda1997calibrating}. The core contribution of the method is to build a fully calibrated model while trying to preserve desirable features of a given model. In effect, we project a given reference model onto the set of calibrated models. As the calibration constraints depend on one-dimensional marginals, classical mimicking results, see \textcite{gyongy1986mimicking} and \textcite{brunick2013mimicking}, allow us to restrict to Markovian models. This in turn allows us to use PDE methods to solve the dual problem. 

The need for a calibrated interest rate model is common to all financial products involving future payments. While in a very low interest rate environment, which roughly held in the financial markets between 2009 and 2021, one may be tempted to ignore this need for short-dated products, it is no longer feasible for the current market conditions. 
It is therefore natural to extend the semimartingale optimal transport approach to a setup that includes stochastic rates.
Our contribution here is to fill this gap and understand how to develop and calibrate joint models for rates and equities. This extension is not trivial, since the computation of the discount factor involves the whole path of the short rate, which renders the problem path dependent. We will show how to overcome this difficulty. We develop suitable duality results and seek to calibrate a stock price local volatility model with volatility coefficient depending on time, the underlying and the short rate process, and driven by a Brownian motion which can be correlated with the randomness driving the rates process. In a sequel paper \textcite{joseph2023joint}, we consider a simultaneous joint calibration problem. A particular difficulty here is in dealing with the path dependent discount terms while keeping the number of state variables, and thus the dimension of the problem, at $d=2$. This is important for computational reasons as the numerical methods rely on solving a non-linear HJB equation and pricing PDEs, which becomes computationally more difficult by standard techniques when the dimension $d\geq 3$. However, our duality result covers any interest rate model, at the expense of harder numerical complexity. 

\section{Preliminaries and Notation}
We adopt the setup of \textcite{guo2022calibration} and \textcite{guo2022joint}, who in turn used the formulation of \textcite{tan2013optimal}. Let $E$ be a Polish space equipped with the Borel $\sigma$-algebra, let $C(E)$ be the space of continuous functions on $E$ and $C_b(E)$ be the space of continuous bounded functions on $E$. Let $\mathcal{M}(E)$ be the space of finite signed Borel measures endowed with the weak-$*$ topology, let $\mathcal{M}_+(E)\subset\mathcal{M}(E)$ be the subset of non-negative finite Borel measures, and $\cP(E)$ be the set of Borel probability measures also under the weak-$*$ topology. Note that if $E$ is compact, then the topological dual of $C_b(E)$ is given by $C_b(E)^*=\mathcal{M}(E)$, but if $E$ is non-compact, then $C_b(E)^*$ is larger than $\mathcal{M}(E)$. Let $\mathrm{BV}(E)$ be the set of bounded variation functions on $E$ and $L^1(\diff\mu)$ be the space of $\mu$-integrable functions. For unambiguity we write $C_b(E;\IR^d)$, $\mathcal{M}(E;\IR^d)$, $\mathrm{BV}(E;\IR^d)$, and $L^1(\diff\mu;\IR^d)$ for the vector valued versions of those spaces (with an analogous definition for the matrix valued versions). 
Write $\IS^d$ for the set of $d\times d$ symmetric matrices and $\IS^d_+\subset\IS^d$ as the subset of positive semidefinite symmetric matrices. For $a,b\in\IR^d$ write $a\cdot b$ for the inner product $a^{\intercal}b$ and for $A,B\in\IS^d$ write $A:B$ for their inner product $\mathrm{Tr}(A^{\intercal}B)$. As a shorthand, we define $\Lambda\coloneqq [0,T]\times\IR^d$ and $\cX\coloneqq\IR\times\IR^d\times\IS^d$, which will be used for the domain and range of the triple representing the law of the semimartingale, the drift and the volatility. Finally, denote the duality bracket between $C_b(E)$ and $C_b(E)^*$ by $\langle\cdot ,\cdot\rangle$.
 
We fix a time horizon $T>0$ and dimension $d>1$. We mainly work on the canonical space $\Omega\coloneqq C([0,T],\IR^d)$ of continuous $\IR^d$-valued paths on $[0,T]$, but sometimes we need to work on $C([0,T],\IR^{d+1})$.
The canonical process is denoted in two ways: either $e$, $e_t(\omega)=\omega_t$, and the dimension is clear from the context, or as $X$, $X_t(\omega)=\omega_t$, with the latter restricted to the $d$-dimensional setting. We take the canonical filtration $\IF=(\cF_t)_{0\leq t\leq T}$ generated by $e$. 
We consider all probability measures $\IP$ on $(\Omega,\cF_T)$ such that $X$ is an $(\IF,\IP)$-semimartingale with decomposition    
\begin{equation*}
X_t=X_0+A_t^{\IP}+M_t^{\IP},\quad t\in[0,T],\quad\IP-\text{a.s.},
\end{equation*}
where $(M_t^{\IP})_{t\geq 0}$ is an $(\IF,\IP)$-martingale and $(A_t^{\IP})_{t\in [0,T]}$ is a finite variation process, both are absolutely continuous relative to the Lebesgue measure and can be characterised in the following sense.
\begin{definition}
We say that $\IP$ is characterised by $(\alpha_t^{\IP},\beta_t^{\IP})$ if
\begin{equation*}
\alpha_t^{\IP}=\frac{\mathrm{d}A_t^{\IP}}{\diff t},\qquad\beta_t^{\IP}=\frac{\mathrm{d}\langle M^{\IP}\rangle_t}{\diff t}\quad \diff t\times \IP(\diff\omega)\text{-a.e.},
\end{equation*}
where $(\alpha_t,\beta_t)_{t\in [0,T]}$ is an $\IR^d\times\IS^d$-valued, progressively measurable process.
\end{definition}
The existence of progressively measurable $(\alpha_t,\beta_t)$ is guaranteed since $A_t^{\IP}$ and $M_t^{\IP}$ are assumed to be Lebesgue absolutely continuous in $t$, $\IP$-a.s., see \textcite[Proposition~I.3.13, Proposition~II.2.9]{jacod2013limit}. 
Note that $(\alpha_t^{\IP},\beta_t^{\IP})$ is only determined up to $\diff\IP\times\diff t$-almost everywhere. The set of probability measures $\IP$ satisfying the conditions above is denoted $\cP$. We note that regular conditional probabilities exist on $\Omega$ and we will use these implicitly, e.g., $\IE^{\IP}_{t,x}[\alpha_t^{\IP}]$ or $\IE^{\IP}[\alpha_t^{\IP}| X_t=x]$ will denote the conditional expectation seen as a measurable function of $(t, X_t)$ and evaluated at $X_t=x$. Finally, we note that, by Doob's martingale representation Theorem (see \textcite[Theorem~3.4.2]{karatzas2014brownian}), possibly on an enlarged probability space, there exists a Brownian motion $W^{\IP}$ such that 
\begin{equation*}
X_t=X_0+\int_0^t\alpha_s^{\IP}\diff s+\int_0^t(\beta_s)^{1/2}\diff W^{\IP}_s,\quad t\leq T.
\end{equation*}

The process $X_t\coloneqq (\mathbf{S}_t,X^r_t)$, is composed of the short rate $(X^r_t)_{t\in [0,T]}$ and a 
$(d-1)$-dimensional process $(\mathbf{S}_t)_{t\in [0,T]}$ corresponding to the underlying asset, such as the S{\&}P 500, and extra state variables, e.g., extra assets, stochastic factors in the volatility functions, or multi-factors in the short rate. 
In specific examples, e.g., in section \ref{sec:2dsetup}, the short rate is also denoted $r_t$. 
The state variable's $x$-coordinate corresponding to $X^r$ is denoted $x_r$ to stress that it represents the short rate.
The stochastic discount factor is derived from the short rate, $Y_t:=\exp(-\int_0^tX^r_s\,\diff s)$, and we refer to $(X_t,Y_t)$ as the augmented process. 
Note that the augmented process can be expressed as an updating function in the sense of \textcite[Definition~3.1]{brunick2013mimicking}.
We further consider the subset $\cP^1\subset\cP$ of measures $\IP\in \cP$ which satisfy
\begin{equation}\label{eq: char integrability}
\IE^{\IP}\Big[\int_0^T\!|\alpha_t^{\IP}|+|\beta_t^{\IP}|\,\diff t\Big]<+\infty\quad\text{and}\quad X^r_t>-1 \diff t\otimes \diff \IP\text{-a.e.,}
\end{equation}
where $|\cdot |$ is the $L^1$ norm. We are not aware of any relevant model which would require the short rate to be unbounded from below so this part of the assumption poses no problem. It implies in particular that $0<Y_t<\mathrm{e}^T$, $t\in [0,T]$, under all measures considered in this paper. 
The integrability assumption would equally be satisfied under any reasonable model and it is imposed so that we can apply \textcite[Theorem~3.6]{brunick2013mimicking}. These mimicking results, extending earlier works of \textcite{krylov1984once} and \textcite{gyongy1986mimicking}, allow us to construct a Markov process with the same one-dimensional marginals as a given semimartingale. As our constraints and cost will only depend on these marginals, it will allow us restrict our attention to Markov processes. Markovian projection relies on the classical result that the marginal law of a diffusion process is a distributional solution to the corresponding Fokker-Planck equation, with the converse result given in \textcite{figalli2008existence} where existence and uniqueness results are constructed for the corresponding SDE satisfied by a process with marginal law that is a weak solution to a Fokker-Planck equation. In addition, since we are considering a calibration problem with a stochastic interest rate, to avoid path dependent terms, we augment the process to include stochastic discount factor terms. For simplicity, our argument will consider a one-factor short rate, but it is easy to adapt to a multi-factor or multi-curve setting. 
\begin{theorem}\label{lem: Markovian projection}
Let $\IP\in \cP^1$ be a candidate model. There exist jointly measurable versions of the conditional expectations, $\alpha_t(x,y)=\IE^{\IP}_{t,x,y}[\alpha_t^{\IP}]$, $\beta_t(x,y)=\IE^{\IP}_{t,x,y}[\beta_t^{\IP}]$, $\diff t\times \diff \IP$-a.e., and $\tilde \IP \in \cP^1$ such that $\alpha^{\tilde \IP}_t = \alpha_t(X_t,Y_t)$ and $\beta_t^{\tilde \IP} = \beta_t(X_t,Y_t)$, $\diff t \times \diff \tilde \IP$ - a.e., and $\tilde \IP\circ (X_t,Y_t)^{-1}=\IP\circ (X_t,Y_t)^{-1}$. \\
Moreover, possibly on some other probability space, there exists a Markov process $(X,Y)$ satisfying
\begin{equation}\label{eq: MP SDE}
\begin{cases}
\diff X_t=\alpha_t(X_t,Y_t)\diff t+[\beta_t(X_t,Y_t)]^{\frac{1}{2}}\diff W_t^{\IP'},&0\leq t\leq T,\\
\diff Y_t=-Y_tX^r_t\diff t,&0\leq t\leq T,\\
X_0=x_0,&\\
Y_0=1.&
\end{cases}
\end{equation}
such that $(X_t,Y_t)\sim \IP\circ (X_t,Y_t)^{-1}$ for all $t\in [0,T]$, and where $W^{\IP'}$ is a Brownian motion.\\
Finally, the marginal distributions, $\rho_t^{}=\rho^{}(t,\cdot)=\IP\circ (X_t,Y_t)^{-1}$, are a weak solution to the Fokker-Planck equation
\begin{equation}\label{eq: MP fokker planck}
\begin{cases}
\partial_t\rho^{}_t+\nabla_x\cdot(\rho_t^{}\alpha_t)-\frac{1}{2}\nabla^2_x :(\rho_t^{}\beta_t) - \partial_y(x_ry\rho^{}_t)=0,&\mathrm{on}\:[0,T]\times\IR^{d+1},\\
\rho_0^{\IP}=\delta_{X_0,Y_0},&(x,y)\in\IR^{d+1}.
\end{cases}
\end{equation}
\end{theorem}
\begin{definition}
 We let $\cP^1_{\mathrm{loc}}$ denote the subset of measures in $\cP^1$ under which the semimartingale characteristics of $X$ are given as measurable functions in $(t,X_t,Y_t)$. 
\end{definition}
\begin{remark}\label{rk:spacesrk}
The first assertion in Theorem \ref{lem: Markovian projection} follows from \textcite[Theorem~7.1]{brunick2013mimicking}. The second one then follows via the martingale representation theorem, and is Theorem~3.6 therein. The last assertion is obtained using It\^o's formula. Note that $\tilde \IP\in  \cP^1_{\mathrm{loc}}$ and is simply the distribution of $X$ solving \eqref{eq: MP SDE}, $\tilde \IP = \IP'\circ X^{-1}$. With a slight abuse of language, we will refer to $\cP^1_{\mathrm{loc}}$ as measures under which $(X,Y)$ is a Markov process solving \eqref{eq: MP SDE}. 
\end{remark}
\begin{remark}
A Borel curve $(\nu_t)_{t\in[0,T]}\subset\cM(\IR^{d})$ is narrowly continuous if for every $f\in C_b(\IR^d)$, the map $t\mapsto \int_{\IR^d}f\diff \nu_t$ is continuous. \textcite[Remark~2.3]{trevisan2016well} (also \textcite[Lemma~8.1.2]{ambrosio2005gradient}) show that any solution $(\rho_t)_{t\in (0,T)}$ of \eqref{eq: MP fokker planck} has a narrowly continuous representative $(\tilde{\rho}_t)_{t\in[0,T]}$ such that $\rho_t=\tilde{\rho}_t$ for almost all $t\in (0,T)$. Define $\mathcal{L}$ as the diffusion operator associated with \eqref{eq: MP fokker planck}, then we also have
$$
\int f_{t_2}\diff\tilde{\rho}_{t_2}-\int f_{t_1}\diff\tilde{\rho}_{t_1}=\int_{t_1}^{t_2}\int\partial_tf+ (\mathcal{L} f)_t\diff\tilde{\rho}_t\diff t,\quad\text{for }0\leq t_1<t_2\leq T.
$$
We therefore may assume without loss of generality that the solution to \eqref{eq: MP fokker planck} is narrowly continuous.
\end{remark}
 Our calibration problem will be written as a minimization of a cost functional which represents a ``distance'' to a given favourite reference model and also ensures perfect calibration. We consider a strongly convex cost function $F:\Lambda\times\IR^d\times\IS^d\to[0,+\infty]$, proper and lower semicontinuous in $(\alpha,\beta)$. It will be set to take value $+\infty$ outside of a set $\Gamma$, see for example \eqref{eq: gamma main case}. In particular we will always set $F=+\infty$ if $\beta\not\in\IS^d_+$ to ensure $\beta$ is a legitimate covariance matrix. We use these properties implicitly, e.g., whenever we restrict to $\beta\in\IS^d_+$. We also define the function $\hat{F}:\Lambda\times\IR\times\IR^d\times\IS^d\to [0,+\infty]$ as $\hat{F}(t,x,y,\alpha,\beta)=f(y)F(t,x,\alpha,\beta)$, where the function $f$ is specified by the discounting term in the payoffs of the options constraints.
 
The strong convexity assumption of $F$ (see \textcite[Definition~2.1.3]{nesterov2018lectures}) means that for any subderivative $\nabla$ performed over $(\alpha,\beta)\in\IR^d\times\IS^d_+$ there exists $C>0$ such that for all $(t,x,\alpha,\beta,\alpha',\beta')\in\Lambda\times\IR^d\times\IS^d_+\times\IR^d\times\IS^d_+$, when $F(t,x,\alpha,\beta) <\infty$ we have
\begin{align}
F(t,x,\alpha',\beta')\geq F(t,x,\alpha,\beta)&+\langle\nabla F(t,x,\alpha,\beta),(\alpha'-\alpha,\beta'-\beta)\rangle\notag\\
&+C(||\alpha-\alpha'||_2^2+||\beta-\beta'||_{\mathrm{Fro}}^2).\label{eq: strong convexity}
\end{align}
Here $||\cdot||_{\mathrm{Fro}}$ denotes the Frobenius norm, which for a matrix $M$ is given by $||M||_{\mathrm{Fro}}=\sqrt{\sum_{i,j}|m_{i,j}|^2}$. We additionally assume that $F$ is $p$-coercive, that is there exists $p>1$ and $C>0$ such that for all $(t,x,\alpha,\beta)\in\Lambda\times\IR^d\times\IS^d_+$ we have
\begin{equation*}
||\alpha||^p+||\beta||^p\leq C[1+F(t,x,\alpha,\beta)].
\end{equation*}
The Legendre-Fenchel transform of $F$ (see for example \textcite[\S 12]{rockafellar1970convex}) is given by 
\begin{equation*}
F^*(t,x,a,b)\coloneqq\sup_{\alpha\in\IR^d,\beta\in\IS_+^d}\{\alpha\cdot a + \beta : b-F(t,x,\alpha,\beta)\},
\end{equation*}
where the supremum is a priori over $(\alpha,\beta)\in \IR^d\times\IS^d$ but can be restricted  $\beta\in\IS^d_+$ by the comment above, or indeed can be later restricted to the set $\Gamma$ as $F=+\infty$ elsewhere. While $F$ itself may not be differentiable with respect to $(\alpha,\beta)$, since $F$ is strictly convex in $(\alpha,\beta)$ we therefore have that $F^*(t,x,a,b)$ is differentiable in $(a,b)$ (see \textcite[Theorem~26.3]{rockafellar1970convex}). For convenience, we will denote $F(\alpha,\beta)\coloneqq F(t,x,\alpha,\beta)$ and $F^*(a,b)\coloneqq F^*(t,x,a,b)$. 

\section{Problem formulation}

We want to calibrate our model to $n$ market prices of options, the $i^{\mathrm{th}}$ option has maturity $\tau_i\in (0,T)$, payoff $G_i\in C_b(\IR^d;\IR)$ and price $u_i$. While European call option payoffs do not satisfy boundedness, as calibrating instruments they are equivalent to put options, via the usual call-put parity, which have bounded and continuous payoffs. We let $\tau=[\tau_1,\dots,\tau_n]$, $G(x)=[G_1(x),\dots,G_n(x)]$ and  $u=[u_1,\dots,u_n]$. We are only interested in modelling on the horizon covered by the market instruments, so we assume that $T=\max_{i\leq n} \tau_i$.

\begin{definition}
Given an initial distribution $\mu_0$, expiry times $\tau$, market prices $u$ corresponding to payoffs $G$, we introduce the set of calibrated measures 
\begin{equation*}
\cP(\mu_0,\tau,u)=\{\IP\in\cP^1 :\:\IP\circ X_{0}^{-1}=\mu_0,\:\IE^{\IP}[Y_{\tau_i}G_i(X_{\tau_i})]=u_i\text{ for }i=1,\dots,n\}.
\end{equation*}
We denote $\cP_{\mathrm{loc}}(\mu_0,\tau,u)=\cP(\mu_0,\tau,u)\cap \cP^1_{\mathrm{loc}}$ the calibrated models under which the semimartingale characteristics of $X$ are given as measurable functions of $(t,X_t,Y_t)$, see Remark \ref{rk:spacesrk}. \\
We denote $\widetilde{\cP}_{\mathrm{loc}}(\mu_0,\tau,u)\subseteq \cP_{\mathrm{loc}}(\mu_0,\tau,u)$ the subset of measures under which the semimartingale characteristics of $X$ are given as measurable functions of $(t,X_t)$. 
\end{definition}
We remark that we will usually take $\mu_0 = \delta_{(X_0,1)}$ where $X_0$ are the observed initial values of our state variables.

Our primal problem consists of selecting one of the possible calibrated models in $\cP(\mu_0,\tau,u)$ by minimising a cost functional. The following key results asserts that when doing so, we can in fact restrict to Markov processes.
\begin{proposition}\label{prop: localised problem}
Given an initial distribution $\mu_0$, expiration times $\tau$, and market prices $u$, we have
\begin{equation}\label{eq: localised problem}
V\coloneqq\inf_{\IP\in\cP(\mu_0,\tau,u)}\IE^{\IP}\bigg[\int_{0}^T Y_tF(\alpha_t^{\IP},\beta_t^{\IP})\,\diff t\bigg]=\inf_{\IP\in\widetilde{\cP}_{\mathrm{loc}}(\mu_0,\tau,u)}\IE^{\IP}\bigg[\int_{0}^T Y_t F(\alpha^{\IP}_t,\beta^{\IP}_t)\,\diff t\bigg].
\end{equation}
\end{proposition}

\subsection{First Markovian reduction}

Similarly to \textcite[Proposition~3.4]{guo2022calibration}, since market constraints only depend on marginal distributions and the cost functional is convex, a combination of Markovian projection in Theorem \ref{lem: Markovian projection} and Jensen's inequality readily shows that we can restrict our attention to Markov process in $(t,X_t,Y_t)$. 
\begin{proof}[of Proposition \ref{prop: localised problem} (Part I)]
Note that $\widetilde{\cP}_{\mathrm{loc}}(\mu_0,\tau,u)\subseteq\cP(\mu_0,\tau,u)$ so the ``$\leq$" is trivial and also if $\cP(\mu_0,\tau,u)$ is empty then both sides are equal $+\infty$. 
We show now the reverse inequality but with $\cP_{\mathrm{loc}}(\mu_0,\tau,u)$ replacing $\widetilde{\cP}_{\mathrm{loc}}(\mu_0,\tau,u)$. 
Take $\IP\in\cP(\mu_0,\tau,u)$ and use Theorem \ref{lem: Markovian projection}, and Remark \ref{rk:spacesrk}, to find the corresponding $\tilde \IP\in\cP_{\mathrm{loc}}(\mu_0,\tau,u)$ such that $(X,Y)$ has the same marginals under $\IP$ and $\tilde \IP$. Using the tower property and Jensen's inequality as $(\alpha,\beta)\mapsto F(\alpha,\beta)$ is convex, we have
\begin{align}
\IE^{\IP}\bigg[\int_{0}^T Y_t F(\alpha_t^{\IP},\beta_t^{\IP})\,\diff t\bigg]&=\IE^{\IP}\bigg[\int_{0}^T Y_t \IE_{t,X_t,Y_t}^{\IP}[F(\alpha_t^{\IP},\beta_t^{\IP})]\,\diff t\bigg]\notag\\
&\geq \IE^{\IP}\bigg[\int_{0}^T Y_t F(\IE^{\mathbb{P}}_{t,X_t,Y_t}[\alpha^{\IP}_t],\IE_{t,X_t,Y_t}^{\IP}[\beta_t^{\IP}])\,\diff t\bigg]\notag\\
&=\IE^{\tilde \IP}\bigg[\int_{0}^T Y_t F(\alpha^{\tilde \IP}_t,\beta^{\tilde \IP}_t)\,\diff t\bigg].\notag
\end{align}
This gives the desired inequality.
\end{proof}
\subsection{A `discounted density' transformation and superposition principle}

While we could easily reduce to Markovian models in the state process $(X,Y)$, this is not satisfactory. In particular, in Section~\ref{sec: numerical results} where we consider a local volatility model with a short rate, we would obtain a three dimensional fully nonlinear PDE instead of a two dimensional one, which involves a substantially increased computational effort. We thus want to further reduce to Markov process $X$, i.e., to obtain \eqref{eq: localised problem}. This requires novel arguments to deal with the stochastic discount term. Instead of working with probability measures, we will work with discounted densities, i.e., sub-probability measures. 
When a $\IP\in \cP^1_{loc}$ is fixed, we write $\rho, \alpha, \beta$, for $\rho^\IP, \alpha^\IP, \beta^\IP$. 
\begin{definition}\label{def: discounting}
Let $\IP\in \cP^1_{\mathrm{loc}}$ with the semimartingale characteristics of $X$ given by $\alpha^{\IP}_t = \alpha_t(X_t,Y_t)$, $\beta_t^{\IP}= \beta_t(X_t,Y_t)$. 
Let $\rho_t(\diff x,\diff y)$ and $\rho_t^X(\diff x)$ be, respectively, the marginal distribution of $(X_t,Y_t)$ and of $X_t$, $t\in [0,T]$. Let $\rho_t(\diff x,\diff y) = \zeta_{t,x}(\diff y) \rho_t^X(\diff x)$, i.e., $\zeta_{t,x}$ is the law of $Y_t$ conditional on $\{X_t=x\}$. 
Define the `discounted' density, drift and volatility for $(t,x)\in[0,T]\times\IR^d$
\begin{align}
\rho^D_t(\diff x)&= D_t(x)\rho_t^X(\diff x),\qquad \textrm{where}\quad D_t(x)=\int_{\IR} y\zeta_{t,x}(\diff y),\label{eq: dd substitution}\\
\alpha^D_t(x)&=\frac{\IE[Y_t\alpha_t(X_t,Y_t)|X_t=x]}{\IE[Y_t|X_t=x]} = \int_{\IR} y\alpha_t(x,y) \frac{\zeta_{t,x}(\diff y)}{D_t(x)},\label{eq: drift substitution}\\
\beta^D_t(x)&=\frac{\IE[Y_t\beta_t(X_t,Y_t)|X_t=x]}{\IE[Y_t|X_t=x]}= \int_{\IR} y\beta_t(x,y) \frac{\zeta_{t,x}(\diff y)}{D_t(x)}.\label{eq: vol substitution}
\end{align}
\end{definition}
By construction, $\rho^D_t$ inherits narrow continuity from $\rho_t$. In addition, using \eqref{eq: char integrability} and \eqref{eq: dd substitution}--\eqref{eq: vol substitution}, we have 
\begin{align}
 \int_0^T\int |\alpha^D(x)| + |\beta^D(x)| \rho^D_t(\diff x)\diff t &=  \int_0^T\int |y\alpha(x,y)| + |y\beta(x,y)| \rho_t(\diff x,\diff y)\diff t \notag\\
 &\leq  e^T\IE^{\IP}[\int_0^T\!|\alpha_t^{\IP}|+|\beta_t^{\IP}|\,\diff t]< +\infty,\label{eq:rhoDintegrability}
\end{align}
 so that $\alpha^D$ and $\beta^D$ are $\diff t \times \rho^D$-integrable. 
 
 A key tool in the dimension reduction is a ``discounted'' version of the superposition principle of \textcite[Theorem~2.5]{trevisan2016well}. We first define the augmented martingale problem.
\begin{definition}
Let $f\in C^{1,2}_b([0,T]\times\IR^{d+1};\IR)$, and given measurable functions $a:[0,T]\times\IR^d\to\IR^d$, $b:[0,T]\times\IR^d\to\IS^d_+$, and $c:[0,T]\times\IR^d\to\IR$, define the operator $$\mathcal{L}f=a\cdot\nabla_xf + \frac{1}{2}b:\nabla^2_xf-yc\partial_yf. $$ 
Then $\boldsymbol{\eta}\in \cP\big(C([0,T];\IR^{d+1})\big)$ is a solution to the augmented martingale problem, $\mathrm{aMP}(a,b,c)$ if
\begin{equation}\label{eq:integrabilityaMP}
\int \int_0^T(|a_t|\circ e_t + |b_t|\circ e_t +|c_t|\circ e_t)\diff t\diff\boldsymbol{\eta}<\infty
\end{equation}
and for all $f\in C^{1,2}_b([0,T]\times\IR^{d+1};\IR)$, the process
\begin{equation*}
[0,T]\ni t\mapsto f(t,\cdot)\circ e_t - \int_0^t[\partial_tf(s,\cdot) + \{\mathcal{L} f\}(s,\cdot)]\circ e_s\diff s
\end{equation*}
is a martingale with respect to the natural filtration on $C([0,T];\IR^{d+1})$.
\end{definition}
We now state the ``discounted superposition principle''.
\begin{theorem}\label{thm: discounted superposition}
Let $(\nu_t)_{t\in [0,T]} \in C\big([0,T];\cM(\IR^d)\big)$ be a narrowly continuous solution of the general discounted Fokker-Planck equation on $[0,T]\times \IR^d$
\begin{equation*}
\partial_t\nu_t+\nabla_x\cdot[a_t(x)\nu_t(x)]-\frac{1}{2}\nabla^2_x:[b_t(x)\nu_t(x)]+c_t(x)\nu_t(x)=0,
\end{equation*}
where $a\in L^1(\diff|\nu_t|\diff t;\IR^d)$, $b\in L^1(\diff|\nu_t|\diff t;\IS^d_+)$, and $c\in L^1(\diff|\nu_t|\diff t;\IR)$. Then there exists $\boldsymbol{\eta}\in \cP(C([0,T];\IR^{d+1}))$ a solution to the $\mathrm{aMP}(a,b,c)$, such that the discounted version of its narrowly continuous curve of marginals $\eta_t=\boldsymbol{\eta}\circ e_t^{-1}$ coincides with $\nu_t$, that is $\int_{\IR}y\eta_t(\cdot,\diff y)=\nu_t$ for all $t\in[0,T]$.
\end{theorem}
The proof of Theorem~\ref{thm: discounted superposition} follows closely the steps in the proof of \textcite[Theorem~2.5]{trevisan2016well}, which itself has a similar structure to \textcite[Theorem~8.2.1]{ambrosio2005gradient}, \textcite[Theorem~4.5]{ambrosio2009flows}, \textcite[Theorem~2.6]{figalli2008existence}, \textcite[Theorem~7.1]{ambrosio2014well}. 
We defer the proof to the Appendix but summarise here the main steps.  
First, in section~\ref{sec: smooth approx}, we argue the result holds for smooth coefficients and constructs smooth approximating sequences in an analogous way to \textcite{trevisan2016well}. Then, Section~\ref{sec: superposition convergence} asserts the compactness of solutions for the augmented martingale problem as a direct consequence of \textcite[Section~A.2]{trevisan2016well}, and establishes convergence via an easy modification of \textcite[Section~A.3]{trevisan2016well}. For the generalisation to smooth and bounded coefficients, the adaptation of the proof is easy with the existence of a solution to $\mathrm{aMP}(a,b,c)$ with a given initial law following directly from the \textcite[Theorem~A.6]{trevisan2016well}, and uniqueness following from a simple adaptation of \textcite[Theorem~A.7]{trevisan2016well}, with the approach in obtaining estimates an easy modification. The generalisation to bounded coefficients relies on the de la Vall\'{e}e Poussin criterion and \textcite[Corollary~A.5]{trevisan2016well}, with our modification again being the handling of the $c$ term --- here we rely on the already established tightness and superposition results for the augmented martingale problem, as it is just the usual martingale problem for $(X,Y)$, and the fact that $Y$ is strictly positive. The estimates for the $x$-coordinates are obtained in the same way as \textcite{trevisan2016well}, and the estimates for the $y$-coordinate use the function $\log(y\chi_R)$ where $\chi_R$ is some cutoff function in order to eliminate the $y$ term in the augmented diffusion operator. The locally bounded coefficients case then follows a similar modified approach to the bounded coefficients case, as in \textcite{trevisan2016well}.
\subsection{Second Markovian reduction}
We use now the tools introduced above to finish the proof of Proposition \ref{prop: localised problem}. 

\begin{proposition}\label{prop:discountedFP}
Let $\IP\in \cP^1_{loc}$ and $\rho^D,\alpha^D,\beta^D$ be given as in Definition~\ref{def: discounting}. Then $\rho^D_t$ is a weak solution to the `discounted' version of the Fokker-Planck equation for $(t,x)\in[0,T]\times\IR^d$
\begin{equation}\label{eq: discount fp}
\partial_t\rho^D_t(x)+\nabla_x\cdot[\alpha^D_t(x)\rho^D_t(x)]-\frac{1}{2}\nabla_x^2:[\beta^D_t(x)\rho^D_t(x)]+x_r\rho^D_t(x)=0. 
\end{equation}
Moreover, there exists a local volatility model $\tilde{\IP}\in\cP^1$ under which the semimartingale characteristics of $X$ are given by $\alpha^{\tilde \IP}_t = \alpha^D_t(X_t)$ and $\beta^{\tilde \IP}_t = \beta^D_t (X_t)$ $\diff t\times \diff \tilde \IP$- a.e. and such that for any payoff function $G\in C_b(\IR^d;\IR)$ and any $t\in[0,T]$
\begin{equation}\label{eq: preserve price}
\IE^{\IP}[Y_tG(X_t)] = \IE^{\tilde{\IP}}[Y_tG(X_t)].
\end{equation}
Furthermore, the value of the objective function decreases:
\begin{equation*}
\IE^{\IP}\Big[\int_0^TY_tF\big(\alpha_t(X_t,Y_t),\beta_t(X_t,Y_t)\big)\diff t\Big]\geq \IE^{\tilde{\IP}}\Big[\int_0^T{Y}_tF\big(\alpha^D_t({X}_t),\beta^D_t({X}_t)\big)\diff t\Big].
\end{equation*}
\end{proposition}
The notion of solution to \eqref{eq: discount fp} is in the spirit of \textcite[Definition~2.2]{trevisan2016well} but with an extra $-x_r f$ term in the diffusion operator of \textcite[Definition~2.1]{trevisan2016well}. 
\begin{proof}[of Proposition \ref{prop: localised problem} (Part II)] Proposition \ref{prop:discountedFP} instantly gives us that 
 \begin{equation*}
\inf_{\IP\in{\cP}_{\mathrm{loc}}(\mu_0,\tau,u)}\IE^{\IP}\bigg[\int_{0}^T Y_tF(\alpha_t^{\IP},\beta_t^{\IP})\,\diff t\bigg]\geq \inf_{\IP\in\widetilde{\cP}_{\mathrm{loc}}(\mu_0,\tau,u)}\IE^{\IP}\bigg[\int_{0}^T Y_tF(\alpha_t^{\IP},\beta_t^{\IP})\,\diff t\bigg],
\end{equation*} 
and hence we have equality since $\widetilde{\cP}_{\mathrm{loc}}(\mu_0,\tau,u)\subseteq {\cP}_{\mathrm{loc}}(\mu_0,\tau,u)$, which completes the proof of Proposition \ref{prop: localised problem}.
\end{proof}

\begin{proof}[of Proposition~\ref{prop:discountedFP}]
Recall the notation of Definition~\ref{def: discounting}. Since $\rho_t$ solves \eqref{eq: MP fokker planck}, we have for $(t,x)\in (0,T)\times\IR^d$
\begin{align*}
\int_{\IR}y\big(\partial_t\rho_t(x,y) + \nabla_x\cdot[\rho_t(x,y)\alpha_t(x,y)]&-\frac{1}{2}\nabla^2_x:[\rho_t(x,y)\beta_t(x,y)]\\&-\partial_y[x_ry\rho_t(x,y)]\big)\diff y=0.
\end{align*}
By construction, the first term is given by $\int_{\IR}y\partial_t\rho_t(x,y)\diff y=\partial_t\rho^D_t(x)$. Then, for the drift we have
\begin{align}
\int_{\IR}y\nabla_x\cdot[\rho_t(x,y)\alpha_t(x,y)]\diff y &= \nabla_x\cdot \Big[\int_{\IR}\alpha_t(x,y)\frac{y}{D_t(x)}\zeta_{t,x}(y)\rho^X_t(x)D_t(x)\diff y\Big],\notag\\
&=\nabla_x\cdot\Big[\frac{\int_{\IR}\alpha_t(x,y)y\zeta_{t,x}(y)\diff y}{D_t(x)}\rho^X_t(x)D_t(x)\Big],\notag\\
&=\nabla_x\cdot[\rho^D_t(x)\alpha^D_t(x)].\notag
\end{align}
The same calculation gives us
\begin{equation*}
\int_{\IR}y\nabla^2_x:[\rho_t(x,y)\beta_t(x,y)]\diff y=\nabla^2_x:[\rho^D_t(x)\beta^D_t(x)].
\end{equation*}
Finally, integration by parts gives $\int_{\IR}y\partial_y[x_ry\rho_t(x,y)]\diff y = -x_r\rho^D_t(x)$, so we have that $\rho^D_t$ solves \eqref{eq: discount fp}. Recall also the integrability condition \eqref{eq:rhoDintegrability} holds. We can thus apply Theorem~\ref{thm: discounted superposition} with $\nu=\rho^D$ to obtain another probability measure, $\tilde{\boldsymbol{\rho}}\in \cP\big(C([0,T];\IR^{d+1})\big)$ solving $\mathrm{aMP}(\alpha^D,\beta^D,x_r)$ such that $\tilde{\rho}^D_t(\cdot) = \int_{\IR}y\tilde{\rho}_t(\cdot,\diff y)=\rho^D(\cdot)$. 
By \textcite[Lemma~3.4]{guo2021path}, for the canonical process $(\tilde X, \tilde Y)$ under $\tilde{\boldsymbol{\rho}}$, 
$\tilde X$ is a local volatility process with the semimartingale characteristics $\alpha^D_t(\tilde X_t)$ and $\beta^{\tilde \IP}_t = \beta^D_t (\tilde X_t)$, and further $\tilde Y_t=\tilde Y_0\exp\big(-\int_0^t \tilde X^r_s\diff s\big)$, $\diff t\times \diff \tilde{\boldsymbol{\rho}}$-a.e. We let $\tilde\IP$ be the projection of $\tilde{\boldsymbol{\rho}}$ on the first $d$ coordinates. Note that $\tilde \IP\in \cP^1$ by \eqref{eq:integrabilityaMP}, $X$ has the desired semimartingale characteristics under $\tilde \IP$ and also $\tilde \rho_t$ is the distribution of $(X_t,Y_t)$ under $\tilde \IP$. 
Then, given a payoff function $G\in C_b(\IR^d;\IR)$, we have for any fixed $t\in[0,T]$
\begin{align*}
\IE^{\IP}[Y_tG(X_t)] = 
\int_{\IR^{d+1}}yG(x)\rho_t(x,y)\diff x\diff y&=\int_{\IR^d}G(x)\rho^D_t(x)\diff x\\
&=\int_{\IR^d}G(x)\tilde{\rho}^D_t(x)\diff x\\
&=\int_{\IR^{d+1}}yG(x)\tilde{\rho}_t(x,y)\diff x\diff y\\
& = \IE^{\tilde{\IP}}[Y_tG(X_t)], 
\end{align*}
so \eqref{eq: preserve price} holds. Finally, we also have 
\begin{align}
\IE^{\IP}\bigg[\int_0^TY_t&F\big(\alpha_t(X_t,Y_t),\beta_t(X_t,Y_t)\big)\diff t\bigg]\notag\\
&=\int_0^T\int_{\IR^{d+1}}yF\big(\alpha_t(x,y),\beta_t(x,y)\big)\zeta_{t,x}(\diff y)\rho^X_t(\diff x)\diff t,\notag\\
&=\int_0^T\int_{\IR^{d+1}}F\big(\alpha_t(x,y),\beta_t(x,y)\big)\frac{y}{D_t(x)}\zeta_{t,x}(\diff y)\,D_t(x)\rho^X_t(\diff x)\diff t,\notag\\
&\geq\int_0^T\int_{\IR^d}F\bigg(\frac{\int_{\IR}y\alpha_t(x,y)\zeta_{t,x}(\diff y)}{D_t(x)},\frac{\int_{\IR}y\beta_t(x,y)\zeta_{t,x}(\diff y)}{D_t(x)}\bigg)D_t(x)\rho^X_t(\diff x)\diff t,\notag\\
&=\int_0^T\int_{\IR^d}F\big(\alpha^D_t(x),\beta^D_t(x)\big)\rho^D_t(\diff x)\diff t,\notag\\
&=\int_0^T\int_{\IR^d}F\big(\alpha^D_t(x),\beta^D_t(x)\big)\tilde{\rho}^D_t(\diff x)\diff t,\label{eq:discdensvalue}\\
&=\IE^{\tilde{\IP}}\bigg[\int_0^T{Y}_tF\big(\alpha^D_t({X}_t),\beta^D_t({X}_t)\big)\diff t\bigg].\label{eq:discdensvaluemodel}
\end{align}
\end{proof}

\section{The dual problem}
Proposition \ref{prop: localised problem} asserted that without any loss of generality, we can restrict our attention to $\IP\in \widetilde{\cP}_{\mathrm{loc}}(\mu_0,\tau,u)$. By conditioning on $X_t$, we can express the value function as an integral against the discounted density, as shown above in \eqref{eq:discdensvalue}-\eqref{eq:discdensvaluemodel}. Accordingly, 
from now on, we always work with the discounted versions of the density, drift, and volatility. We drop the superscript $D$ for notational ease. The problem we need to solve in the RHS of \eqref{eq: localised problem} is thus given equivalently as follows:
\begin{problem}\label{prob: primal}
The value function for the Primal Problem is given by
\begin{equation}\label{eq: objective function}
V:=\inf_{\rho,\alpha,\beta}\int_0^T\int_{\IR^d}\!F\big(\alpha_t(x),\beta_t(x)\big)\rho_t(\diff x)\diff t,
\end{equation}
where the infimum is taken over $(\rho,\alpha,\beta)\in C\big([0,T];\cM(\IR^d)\big)\times L^1(\diff\rho_t\diff t;\IR^d)\times L^1(\diff\rho_t\diff t;\IS^d)$, subject to the following constraints in the sense of distributions for $(t,x)\in [0,T]\times\IR^d$
\begin{align*}
\partial_t\rho_t(x)+\nabla_x[\rho_t(x)\alpha_t(x)]-\frac{1}{2}\nabla^2_x:[\rho_t(x)\beta_t(x)]+x_r\rho_t(x)&=0,\\
\text{for }i=1,\dots,n,\quad\int_{\IR^d}\!G_i(x)\rho_{\tau_i}(\diff x)&=u_i,\\
\rho_0(\cdot)&=\mu_0.
\end{align*}
\end{problem}
We note that Problem~\ref{prob: primal} is in fact an example of unbalanced optimal transport, with the source term in the continuity equation being $-x_r\rho$, see \textcite{chizat2018unbalanced,sejourne2023unbalanced}. To solve this constrained optimisation problem, we will use a duality method inspired by \textcite{huesmann2019benamou, guo2022calibration, guo2022joint}. This proof is similar in spirit to the duality argument used in \textcite[Proposition~2.7]{brenier1999minimal}, which considers dual of a relaxation of the minimal geodesic problem as an alternative approach to solving the Euler equations of an incompressible fluid. This technique had previously been applied in the proof of \textcite[Theorem~3.2]{brenier1997homogenized} when seeking for the existence variational solutions of ``homogenised vortex sheet equations''. The technique is also used to formulate the dual of a variational problem involving the Euler-Poisson system of equations in \textcite{loeper2006reconstruction}. The proof relies mainly on the Fenchel-Rockafellar duality theorem (see \textcite[Chapter~31]{rockafellar1970convex}), and an adjustment to make the problem convex. As our primal problem is quite similar to \textcite{guo2022calibration}, the approach used there can be adapted to our setting. The following result uses the notion of viscosity solution from Definition \ref{def: viscosity solution} below.
\begin{theorem}\label{thm: dual problem}
The dual expression for the value function $V$ is
\begin{equation}\label{eq: dual problem}
V=\sup_{\lambda}\Big\{\lambda\cdot u-\int_{\IR^d}\!\phi^\lambda(0,x)\,\diff\mu_0\Big\},
\end{equation}
where $\lambda\in\IR^n$ and $\phi^\lambda=\phi$ is the viscosity solution to the HJB equation
\begin{equation}\label{eq: dual HJB}
\partial_t\phi-x_r\phi+\sum_{i=1}^n\lambda_iG_i(x)\delta_{\tau_i}+F^*\bigg(\nabla_x\phi,\frac{1}{2}\nabla_x^2\phi\bigg)=0,\quad (t,x)\in [0,T]\times\IR^d,
\end{equation}
with terminal condition $\phi(T,\cdot)=0$. If $V$ is finite, then the infimum in Problem~\ref{prob: primal} is attained. If the supremum is attained for some $\lambda^*\in\IR^n$, with $\phi^*\in\mathrm{BV}([0,T];C_b^2(\IR^d))$ such that $\phi^*(T,\cdot) = 0$ solving the corresponding HJB equation, and $(\rho^*,\alpha^*,\beta^*)$ being the optimal solution of Problem~\ref{prob: primal}, then $(\alpha^*,\beta^*)$ is given by
\begin{equation*}
(\alpha^*_t,\beta^*_t)=\nabla F^*\bigg(\nabla_x\phi^*(t,\cdot),\frac{1}{2}\nabla_x^2\phi^*(t,\cdot)\bigg),\qquad\mathrm{d}\rho^*_t\diff t\text{ - almost everywhere}.
\end{equation*}
\end{theorem}
\begin{remark}
Solving the dual problem in Theorem~\ref{thm: dual problem} yields the primal optimisers $(\alpha^*,\beta^*)$. These characterise the distribution of the Markov process $(X,Y)$ solving \eqref{eq: MP SDE} with these coefficients , i.e., $\rho^*$ solves \eqref{eq: MP fokker planck}, and the associated measure  $\IP^*$ on $\Omega$,  $\rho^*_t=\IP^*\circ (X_t,Y_t)^{-1}$. We can then compute the associated discount factor $D_t^*(x)$ and apply the transformation \eqref{eq: dd substitution} to obtain that $D_t^*(x)\rho_t^*(x)$ is a solution of \eqref{eq: discount fp}. 
\end{remark}
\begin{remark}
It is not presently known if the supremum in \eqref{eq: dual problem} is attained for a generic cost function $F$ --- in general, this is an open problem in (semi)martingale optimal transport. In the static case of martingale OT, \textcite{beiglbock2019dual} show that dual attainment is delicate and depends on convexity and growth properties of the cost in the second marginal. In the Bass martingale setting of \textcite[Section~7]{mBBRd1}, that is under the assumption of irreducibility of marginals, the authors establish the attainability of the static dual corresponding to a specific weak martingale OT problem. This static problem is equivalent to the dynamic formulation of \textcite{tan2013optimal,huesmann2019benamou}, see also \textcite{joseph2023measure}.
\end{remark}
We prove the duality result above in the remainder of this section through a series of lemmas. Our first observation is that the objective function \eqref{eq: objective function} is not jointly convex in $(\rho,\alpha,\beta)$. We define the measures $\cA\coloneqq \rho\alpha$, $\cB\coloneqq\rho\beta$, so $\cA$ and $\cB$ are absolutely continuous with respect to $\rho$. Then, the objective function is convex in $(\rho,\cA,\cB)$ with constraints that are affine in $(\rho,\cA,\cB)$. This arises from the classical notion that the function $\bar{f}(z_1,z_2,z_3)\coloneqq z_3f(\frac{z_1}{z_3},\frac{z_2}{z_3})$ is convex in $(z_1,z_2,z_3)$ whenever $f$ is convex in $(z_1,z_2)$ on the set $\{z_3>0\}$. Note also that we write $\diff \rho$ for $\rho_t(\diff x)\diff t$, $\diff\cA$ for $\alpha_t(x)\rho_t(\diff x)\diff t$ and $\diff\cB$ for $\beta_t(x)\rho_t(\diff x)\diff t$. Moreover, our constraints in Problem~\ref{prob: primal} can be formulated in the weak sense as:
\begin{align}
\int_{\Lambda}\!\partial_t\phi\mathrm{d}\rho+\nabla\phi\cdot\mathrm{d}\cA+\frac{1}{2}\nabla^2\phi :\mathrm{d}\cB-x_r\phi\mathrm{d}\rho+\int_{\IR^d}\phi\mathrm{d}\mu_0&=0,\label{eq: weak constraints}\\
\int_{\Lambda}\sum_{i=1}^n\lambda_iG_i(x)\delta_{\tau_i}\mathrm{d}\rho-\sum_{i=1}^n\lambda_iu_i&=0.\notag
\end{align}
\noindent for any smooth compactly supported test function $\phi\in C_c^{\infty}(\Lambda)$ with $\phi(T,\cdot) = 0$ and $\lambda\in\IR^n$. The terminal condition on $\phi$ arises when integrating by parts to derive \eqref{eq: weak constraints}, since we need the $\rho_T(\cdot)$ boundary term to vanish as we do not have a priori knowledge of $\rho_T(\cdot)$. Therefore we can write Problem~\ref{prob: primal} as the following saddle point problem
\begin{problem}\label{prob: saddle pt}
\begin{align}
V=\inf_{\rho,\cA,\cB}\sup_{\phi,\lambda}\bigg\{&\int_{\Lambda}F\bigg(\frac{\mathrm{d}\cA}{\mathrm{d}\rho},\frac{\mathrm{d}\cB}{\mathrm{d}\rho}\bigg)\mathrm{d}\rho-\partial_t\phi\mathrm{d}\rho-\nabla\phi\cdot\mathrm{d}\cA-\frac{1}{2}\nabla^2\phi :\mathrm{d}\cB\notag\\
&+x_r\phi\mathrm{d}\rho-\int_{\IR^d}\phi\mathrm{d}\mu_0-\int_{\Lambda}\sum_{i=1}^n\lambda_iG_i(x)\delta_{\tau_i}\mathrm{d}\rho+\sum_{i=1}^n\lambda_iu_i\bigg\},\notag
\end{align}
where the infimum is taken over $(\rho,\alpha,\beta)\in C\big([0,T];\cM(\IR^d)\big)\times L^1(\diff\rho_t\diff t;\IR^d)\times L^1(\diff\rho_t\diff t;\IS^d)$, $\cA= \rho\alpha$, $\cB=\rho\beta$,  and the supremum over $(\phi,\lambda)\in C_c^{\infty}(\Lambda;\IR)\times\IR^n$.
\end{problem}
\noindent We now want to find a functional with convex conjugate equal to \eqref{eq: objective function}, and another that is the remainder of the infimum in Problem~\ref{prob: saddle pt}. To do this, we use the following terminology from \textcite{huesmann2019benamou} in the proof of the duality theorem. Denote $\mathrm{BV}_T\big([0,T];C^2_b(\IR^d)\big)$ as the set of $\phi\in\mathrm{BV}\big([0,T];C^2_b(\IR^d)\big)$ such that $\phi(T,\cdot)=0$.
\begin{definition}\label{def:represent}
We say that the triple $(\gamma,a,b)\in C_b(\Lambda;\cX)$ is represented by $(\phi,\lambda)\in\mathrm{BV}_T\big([0,T];C^2_b(\IR^d)\big)\times\IR^n$ if 
\begin{align*}
\gamma + \partial_t\phi -x_r\phi + \sum_{i=1}^n\lambda_iG_i(x)\delta_{\tau_i}&=0,\\
a+\nabla\phi&=0,\\
b+\frac{1}{2}\nabla^2\phi&=0.
\end{align*}
\end{definition}
\noindent Since $(\gamma,a,b)\in C_b(\Lambda;\cX)$, the presence of the dirac delta functions give that $t\mapsto\phi(t,\cdot)$ is of bounded variation on $[0,T]$ with jump discontinuities at $t=\tau_i$, which we denote as $\phi\in\mathrm{BV}_T\big([0,T];C^2_b(\IR^d)\big)$, since we require $\phi$ to be at least $C^2$ in space. Proceding in an analogous way to \textcite{guo2022calibration} to obtain the dual problem, first define functionals $\Phi :C_b(\Lambda;\cX)\to\IR\cup\{+\infty\}$ and $\Psi:C_b(\Lambda;\cX)\to\IR\cup\{+\infty\}$ by
\begin{eqnarray*}
\Phi(\gamma,a,b)&=&\begin{cases}
0,&\text{if $\gamma + F^*(a,b)\leq 0$,}\\
+\infty,&\text{otherwise}.
\end{cases}\\
\Psi(\gamma,a,b)&=&\begin{cases}
\int_{\IR^d}\phi(0,x)\,\mathrm{d}\mu_0-\sum_{i=1}^n\lambda_iu_i,&\begin{aligned}\!&\text{if $(\gamma,a,b)$ is represented by}(\phi,\lambda)\in\\
&\mathrm{BV}_T\big([0,T];C^2_b(\IR^d)\big)\times\IR^n,\end{aligned}\\
+\infty,&\mathrm{otherwise}.
\end{cases}
\end{eqnarray*}
\begin{lemma}
The objective function $V$ can be expressed in terms of $\Phi$ and $\Psi$ as
\begin{equation*}
V=\inf_{\rho,\cA,\cB}\{\Phi^*(\rho,\cA,\cB)+\Psi^*(\rho,\cA,\cB)\},
\end{equation*}
where the infimum is taken across $(\rho,\cA,\cB)\in C_b(\Lambda,\cX)^*$.
\end{lemma}
\sloppy
We remark that switching from $\phi\in C_c^{\infty}(\Lambda)$ (with $\phi(T,\cdot)=0$) to $\phi\in\mathrm{BV}_T\big([0,T];C^2_b(\IR^d)\big)$ does not change the value of the supremum, which will be formalised by the notion of viscosity solutions later in Definition~\ref{def: viscosity solution}
\begin{proof}
We first evaluate $\Phi^*$ on $(\rho,\cA,\cB)\in\mathcal{M}(\Lambda;\cX)$ such that $\rho\in\mathcal{M}_+(\Lambda)$ and $\cA,\cB\ll\rho$ as follows
\begin{align}
\Phi^*(\rho,\cA,\cB)&=\sup_{\gamma+F^*(a,b)\leq 0}\int_{\Lambda}\!\bigg(\gamma+a\cdot\frac{\mathrm{d}\cA}{\mathrm{d}\rho}+b:\frac{\mathrm{d}\cB}{\mathrm{d}\rho}\bigg)\,\mathrm{d}\rho\notag\\
&=\sup_{a,b}\int_{\Lambda}\!\bigg(a\cdot\frac{\mathrm{d}\cA}{\mathrm{d}\rho}+b:\frac{\mathrm{d}\cB}{\mathrm{d}\rho}-F^*(a,b)\bigg)\,\mathrm{d}\rho\notag\\
&=\int_{\Lambda}\!\sup_{a,b}\bigg(a\cdot\frac{\mathrm{d}\cA}{\mathrm{d}\rho}+b:\frac{\mathrm{d}\cB}{\mathrm{d}\rho}-F^*(a,b)\bigg)\,\mathrm{d}\rho =\int_{\Lambda}\!F\bigg(\frac{\mathrm{d}\cA}{\mathrm{d}\rho},\frac{\mathrm{d}\cB}{\mathrm{d}\rho}\bigg)\,\mathrm{d}\rho,\notag
\end{align}
where the exchange of the integral and the supremum is justified as in \textcite[Lemma A.1]{guo2022calibration}. Furthermore, the arguments in \textcite[Lemma A.1]{guo2022calibration} also show that $\Phi^*=\infty$ for 
$(\rho,\cA,\cB)\in C_b(\Lambda;\cX)^*$ which are not of the form above. In summary:
\begin{equation}\label{eq: phi conjugate}
\Phi^*(\rho,\cA,\cB)=\begin{cases}\int_{\Lambda}\!F\bigg(\frac{\mathrm{d}\cA}{\mathrm{d}\rho},\frac{\mathrm{d}\cB}{\mathrm{d}\rho}\bigg)\,\mathrm{d}\rho&\begin{aligned}\!&\text{if } (\rho,\cA,\cB)\in\mathcal{M}(\Lambda,\cX),\\& \rho\in \mathcal{M}_+(\Lambda), \cA,\cB\ll\rho\end{aligned}\\+\infty&\text{otherwise}.\end{cases}
\end{equation}
We now compute $\Psi^*:C_b(\Lambda;\cX)^*\to\IR\cup\{+\infty\}$:
\begin{align}
\Psi^*(\rho,&\cA,\cB)=\sup_{(\gamma,a,b)\in C_b(\Lambda;\cX)}\Big\{\langle (\gamma,a,b),(\rho,\cA,\cB)\rangle-\int_{\IR^d}\!\phi(0,x)\mathrm{d}\mu_0+\sum_{i=1}^n\lambda_iu_i\Big\}\notag\\
&=\sup_{\phi,\lambda}\bigg\{\bigg\langle \bigg(x_r\phi-\partial_t\phi - \sum_{i=1}^n\lambda_iG_i(x)\delta_{\tau_i},-\nabla\phi,-\frac{1}{2}\nabla^2\phi\bigg),(\rho,\cA,\cB)\bigg\rangle\notag\\
&\qquad\qquad\qquad\qquad-\int_{\IR^d}\!\phi(0,x)\mathrm{d}\mu_0+\sum_{i=1}^n\lambda_iu_i\bigg\}\notag\\
&=\sup_{\phi,\lambda}\bigg\{\bigg[\int_{\Lambda}\!x_r\phi\mathrm{d}\rho-\partial_t\phi\mathrm{d}\rho- \sum_{i=1}^n\lambda_iG_i(x)\delta_{\tau_i}\mathrm{d}\rho-\nabla_x\phi\mathrm{d}\cA-\frac{1}{2}\nabla_x^2\phi:\mathrm{d}\cB\bigg]\notag\\
&\qquad\qquad\qquad\qquad-\int_{\IR^d}\!\phi(0,x)\mathrm{d}\mu_0+\sum_{i=1}^n\lambda_iu_i\bigg\},\label{eq: psi conjugate}
\end{align}
where $(\gamma,a,b)$ are represented by $(\phi,\lambda)$ in the sense of Definition \ref{def:represent}.

Thus, since $\Phi^*$ is independent of $(\phi,\lambda)$ we can simply add \eqref{eq: phi conjugate} and \eqref{eq: psi conjugate} together to get the argument of the infimum in Problem~\ref{prob: saddle pt}, so
\begin{equation*}
V=\inf_{(\rho,\cA,\cB)\in \mathcal{M}(\Lambda;\cX)}\{\Phi^*(\rho,\cA,\cB)+\Psi^*(\rho,\cA,\cB)\},
\end{equation*}
where the restriction $\rho\in \mathcal{M}_+(\Lambda), \cA,\cB\ll\rho$ are ensured automatically since otherwise $\Phi^*=+\infty$. To finish the proof, we invoke \textcite[Lemma A.2]{guo2022calibration} to conclude that  
\begin{align*}
V&=\inf_{(\rho,\cA,\cB)\in\mathcal{M}(\Lambda;\cX)}\{\Phi^*(\rho,\cA,\cB)+\Psi^*(\rho,\cA,\cB)\}\\&=\inf_{(\rho,\cA,\cB)\in C_b(\Lambda;\cX)^*}\{\Phi^*(\rho,\cA,\cB)+\Psi^*(\rho,\cA,\cB)\}.
\end{align*}
\end{proof}

Now we apply the Fenchel-Rockerfellar duality theorem as formulated in \textcite[Theorem~1.9]{villani2003topics}. We first note that as the constraints in the functionals $\Phi$ and $\Psi$ are affine, the functionals are clearly convex in $(\gamma,a,b)$. We now check the conditions of the theorem at the point $(0,O^{d\times 1},O^{d\times d})$ which is represented by $(\phi,\lambda)=(e^{tx_r},O^{n\times 1})$ where $O$ refers to the zero matrix of appropriate dimension. Since $F$ is non-negative
\begin{equation*}
F^*(O^{d\times 1})=-\inf_{(\alpha,\beta)\in\IR^d\times\IS^d_+}F(\alpha,\beta)\leq 0.
\end{equation*}
Therefore, we have $\Phi(0,O^{d\times 1},O^{d\times d})=0$ and $(0,O^{d\times 1},O^{d\times d})$ is a point of continuity of $\Phi$ since $F^*$ is continuous. Moreover, since $\phi(t,x)=e^{tx_r}$ and $\mu_0$ is a probability measure
\begin{equation*}
\Psi(0,O^{d\times 1},O^{d\times d})=\int_{\IR^d}\!\phi(0,x)\,\mathrm{d}\mu_0=1.
\end{equation*}
Thus, as $\Psi$ is finite and $\Phi$ is finite and continuous at $(0,O^{d\times 1},O^{d\times d})$ and both take values in $(-\infty,+\infty]$ we may apply the Fenchel-Rockerfellar duality theorem and obtain that
\begin{align}
\inf_{(\gamma,a,b)\in C_b(\Lambda;\cX)}&\{\Phi(-\gamma,-a,-b)+\Psi(\gamma,a,b)\}\notag\\
&=\sup_{(\rho,\cA,\cB)\in C_b(\Lambda;\cX)^*}\{-\Phi^*(-\rho,-\cA,-\cB)-\Psi^*(-\rho,-\cA,-\cB)\}\notag\\
&=\sup_{(\rho,\cA,\cB)\in C_b(\Lambda;\cX)^*}\{-\Phi^*(\rho,\cA,\cB)-\Psi^*(\rho,\cA,\cB)\}\notag\\
&=-\inf_{(\rho,\cA,\cB)\in C_b(\Lambda;\cX)^*}\{\Phi^*(\rho,\cA,\cB)+\Psi^*(\rho,\cA,\cB)\}.\notag
\end{align}
Thus rearranging we obtain
\begin{equation}\label{eq:V super sol}
\begin{split}
V&=\inf_{(\rho,\cA,\cB)\in C_b(\Lambda;\cX)^*}\{\Phi^*(\rho,\cA,\cB)+\Psi^*(\rho,\cA,\cB)\}\\&=\sup_{(\gamma,a,b)\in C_b(\Lambda;\cX)}\{-\Phi(-\gamma,-a,-b)-\Psi(\gamma,a,b)\}\\
&=\sup_{\substack{(\gamma,a,b)\in C_b(\Lambda;\cX)\\-\gamma+F^*(-a,-b)\leq 0\\ (\gamma,a,b)\text{is represented by } (\phi,\lambda)}}\bigg\{\sum_{i=1}^n\lambda_iu_i-\int_{\IR^d}\!\phi(0,x)\,\mathrm{d}\mu_0\bigg\}\\
&=\sup_{\substack{(\lambda,\phi)\in\IR^n\times\mathrm{BV}_T([0,T];C^2_b(\IR^d))\\\partial_t\phi -x_r\phi + \sum_{i=1}^n\lambda_iG_i(x)\delta_{\tau_i}+F^*(\nabla_x\phi,\frac{1}{2}\nabla^2_x\phi)\leq 0}}\bigg\{\sum_{i=1}^n\lambda_iu_i-\int_{\IR^d}\!\phi(0,x)\,\mathrm{d}\mu_0\bigg\}.
\end{split}
\end{equation}
Now, to obtain equality in the HJB equation constraint and thus the HJB equation \eqref{eq: dual HJB} and dual formulation in Theorem~\ref{thm: dual problem}, we adapt the classical notion of viscosity solutions from \textcite{lions1983optimal} to include the required jump discontinuities, in analogy to \textcite{guo2022calibration}. First define disjoint intervals $I_k\coloneqq [\tau_{k-1},\tau_k)$ with $\tau_0=0$, with $\bigcup_{k=1}^n I_k=[0,T)$. 
\begin{definition}\label{def: viscosity solution}
For any $\lambda\in\IR^n$, we say $\phi\in\mathrm{BV}_T([0,T];C_b(\IR^d))$ is a viscosity subsolution (supersolution) of \eqref{eq: dual HJB} if $\phi|_{I_k\times{\IR^d}}\in C_b(I_k;C_b(\IR^d))$  is a classical (continuous) viscosity subsolution (supersolution) of \eqref{eq: dual HJB} in $I_k\times\IR^d$ and for all $k=1,\dots,n$ has jump discontinuities:
\begin{equation*}
\phi(t,x)=\phi(t^-,x)-\sum_{i=1}^n\lambda_iG_i(x)\indic{\{t=\tau_i\}}.
\end{equation*}
With terminal condition $\phi(T,\cdot) = 0$. In addition, $\phi\in\mathrm{BV}_T([0,T];C_b(\IR^d))$ is a viscosity solution of \eqref{eq: dual HJB} if it is both a viscosity subsolution and viscosity supersolution of \eqref{eq: dual HJB}.
\end{definition}

The expression for $V$ in \eqref{eq:V super sol} involved supersolutions to \eqref{eq: dual HJB} and the first step is to show that we can restrict to viscosity solutions, as stated in the first part of Theorem ~\ref{thm: dual problem}. For this we follow the proof of \textcite[Proposition~3.5]{guo2022calibration}. \textcite[Remark~3.9]{guo2022calibration} provides a comparison principle, which is the classical comparison principle of viscosity solutions applied to $\phi$ on $I_k\times\IR^d$ for each $k$. Using this, one can deduce existence and uniqueness of solutions to \eqref{eq: dual HJB} via \textcite{crandall1992user}.
The comparison principle also implies that $V$ in \eqref{eq: dual problem} is smaller than the supremum over viscosity solutions. We then use the smoothing argument from \textcite{bouchard2017hedging}, which shows that any viscosity solution of \eqref{eq: dual HJB} can be approached by smooth supersolutions. This, together with \eqref{eq:V super sol} shows that $V$ is larger than the supremum over viscosity solutions. This allows us to conclude the proof of the first point of Theorem~\ref{thm: dual problem}. We now seek to obtain the form of the optimal $(\alpha,\beta)$ for the second part of Theorem~\ref{thm: dual problem}.
\begin{lemma}\label{lem: optimisers}
If the supremum in Theorem~\ref{thm: dual problem} is attained for some $\lambda^*$ with $\phi^*\in\mathrm{BV}_T([0,T];C^2_b(\IR^d))$ solving the corresponding HJB equation, and $(\rho^*,\alpha^*,\beta^*)$ being the optimal solution of Problem~\ref{prob: primal}, then $(\alpha^*,\beta^*)$ is given by:
\begin{equation*}
(\alpha^*_t,\beta^*_t)=\nabla F^*\bigg(\nabla_x\phi^*(t,\cdot),\frac{1}{2}\nabla_x^2\phi^*(t,\cdot)\bigg),\qquad\mathrm{d}\rho^*\text{ - almost everywhere}.
\end{equation*}
\end{lemma} 
\begin{proof}
Let $(\alpha^*,\beta^*)$ be the optimal solution of Problem~\ref{prob: primal}, then $(\rho^*,\rho^*\alpha^*,\rho^*\beta^*)$ also achieves the infimum in Problem~\ref{prob: saddle pt}. Assume that $\lambda^*$ is the optimal solution solving \eqref{eq: dual problem} with corresponding solution to \eqref{eq: dual HJB} $\phi^*$. Then, $\lambda^*$ achieves the supremum in Problem~\ref{prob: saddle pt}, so with our optimal solution we may rewrite Problem~\ref{prob: saddle pt} as
\begin{align}
V=\int_{\Lambda}\biggl(F(\alpha^*,\beta^*)-\partial_t\phi^*&-\nabla_x\phi^*\cdot\alpha^*-\frac{1}{2}\nabla^2_x\phi^* :\beta^*+x_r\phi^*\notag\\
&-\sum_{i=1}^n\lambda^*G_i(x)\delta_{\tau_i}\biggr)\,\mathrm{d}\rho^*-\int_{\IR^d}\!\phi^*\,\mathrm{d}\mu_0+\sum_{i=1}^n\lambda_i^*u_i.\label{eq: optimal saddle pt}
\end{align}
Since $(\phi^*,\lambda^*)$ are optimal, we have from Theorem~\ref{thm: dual problem} that 
\begin{equation*}
V=\sum_{i=1}^n\lambda_i^*u_i-\int_{\IR^d}\!\phi^*\,\mathrm{d}\mu_0.
\end{equation*}
Therefore, \eqref{eq: optimal saddle pt} is equivalent to
\begin{align}
0&=\int_{\Lambda}\!\bigg(F(\alpha^*,\beta^*)-\partial_t\phi^*-\nabla_x\phi^*\cdot\alpha^*-\frac{1}{2}\nabla^2_x\phi^* :\beta^*+x_r\phi^*-\sum_{i=1}^n\lambda^*G_i(x)\delta_{\tau_i}\bigg)\diff\rho^*\notag\\
&=\int_{\Lambda}\!\bigg(F(\alpha^*,\beta^*)+F^*(\nabla_x\phi^*,\frac{1}{2}\nabla^2_x\phi^*)-\nabla_x\phi^*\cdot\alpha^*-\frac{1}{2}\nabla^2_x\phi^* :\beta^*\bigg)\,\mathrm{d}\rho^*.\label{eq: F and F star}
\end{align}
Now define $(\bar{\alpha},\bar{\beta})$ as
\begin{equation*}
(\bar{\alpha},\bar{\beta})=\nabla F^*\bigg(\nabla_x\phi^*,\frac{1}{2}\nabla^2_x\phi^*\bigg).
\end{equation*}
Note that 
\begin{equation}\label{eq: argmax legendre}
(\bar{\alpha},\bar{\beta})=\nabla F^*\bigg(\nabla_x\phi^*,\frac{1}{2}\nabla_x^2\phi^*\bigg)=\argmax_{(a,b)\in\IR^d\times\IS^d_+}\bigg\{a\cdot\nabla_x\phi^*+b:\frac{1}{2}\nabla^2_x\phi^*-F(a,b)\bigg\}.
\end{equation}
Thus,
\begin{equation*}
F^*\bigg(\nabla_x\phi^*,\frac{1}{2}\nabla_x^2\phi^*\bigg)=\nabla_x\phi^*\cdot\bar{\alpha}+\frac{1}{2}\nabla_x^2\phi^* :\bar{\beta}-F(\bar{\alpha},\bar{\beta}).
\end{equation*}
Since $F$ is convex, its Legendre transform is an involution (see \textcite[Corollary~12.2.1]{rockafellar1970convex}), so taking the double Legendre transform, we have
\begin{equation*}
(A,B)\coloneqq\nabla F(\bar{\alpha},\bar{\beta})=\nabla F^{**}(\bar{\alpha},\bar{\beta})=\argmax_{(a,b)\in\IR^d\times\IS^d_+}\{a\cdot\bar{\alpha}+b:\bar{\beta}-F^*(a,b)\}.
\end{equation*}
Therefore,
\begin{align*}
F^{**}(\bar{\alpha},\bar{\beta})&=A\cdot\bar{\alpha}+B:\bar{\beta}-F^*(A,B)\\
&=A\cdot\bar{\alpha}+B:\bar{\beta}-\max_{(x,y)\in\IR^d\times\IS^d_+}\{A\cdot x+B:y-F(x,y)\}.
\end{align*}
Since $F(\bar{\alpha},\bar{\beta})=F^{**}(\bar{\alpha},\bar{\beta})$, from \eqref{eq: argmax legendre} we have that for the affine terms to cancel, we need $(A,B)=(\nabla_x\phi^*,\frac{1}{2}\nabla^2_x\phi^*)$. So, substituting into \eqref{eq: F and F star} we have
\begin{equation}\label{eq: alpha star and alpha bar}
0=\int_{\Lambda}\!\big(F(\alpha^*,\beta^*)-F(\bar{\alpha},\bar{\beta})-\nabla_x\phi^*\cdot(\alpha^*-\bar{\alpha})-\frac{1}{2}\nabla_x^2\phi^* :(\beta^*-\bar{\beta})\big)\,\mathrm{d}\rho^*.
\end{equation}
Since $F(\alpha,\beta)$ is assumed to be strongly convex in $(\alpha,\beta)$, we have from \eqref{eq: strong convexity} that for some constant $C>0$
\begin{equation*}
F(\alpha^*,\beta^*)-F(\bar{\alpha},\bar{\beta})\geq\langle\nabla F(\bar{\alpha},\bar{\beta}),(\alpha^*-\bar{\alpha},\beta^*-\bar{\beta})\rangle+C\big(||\alpha^*-\bar{\alpha}||^2+||\beta^*-\bar{\beta}||^2\big).
\end{equation*}
Applying this inequality to \eqref{eq: alpha star and alpha bar} and noting that $\nabla F(\bar{\alpha},\bar{\beta})=(\nabla_x\phi^*,\frac{1}{2}\nabla^2_x\phi^*)$ gives us
\begin{equation*}
0\geq\int_{\Lambda}\!C(||\alpha^*-\bar{\alpha}||^2+||\beta^*-\bar{\beta}||^2)\,\mathrm{d}\rho^*\geq 0.
\end{equation*}
Therefore we have $(\alpha^*,\beta^*)=(\bar{\alpha},\bar{\beta})$ up to $\mathrm{d}\rho^*$ almost everywhere.
\end{proof}

\subsection{Sequential Calibration Setup}\label{sec:2dsetup}
We now specify the setting in which we seek to apply Theorem~\ref{thm: dual problem}. We first specify the state variables of our model: we consider $d=2$, $X_t = (r_t, Z_t)$, were $r_t$ is the short rate and $Z_t$ is the log-stock price. We start with a setting in which a model for the short rate is fixed and has already been calibrated. We refer to this as a ``sequential calibration'' problem. Our aim is then to calibrate a local volatility model for the stock price, where the local volatility function can depend on the short rate. For this we will employ the OT methodology developed above. The joint calibration problem in which both the short rate and the stock price are calibrated simultaneously using the OT methodology, and in particular the short rate model parameters can depend on the stock price, is considered in our sequel paper \textcite{joseph2023joint}.

Specifically, we take now a given pre-calibrated Hull-White model for the interest rate, see \textcite{hull1990pricing, hull1994branching}, and local volatility dynamics for the log-price
\begin{align*}
\diff Z_t &= r_t - \frac{1}{2}\sigma^2(t,Z_t,r_t)+\sigma(t,Z_t,r_t)\diff W^1_t,\\
\diff r_t &= \big(b(t)-a r_t\big)\diff t +\sigma_r\diff W^2_t,\\
\diff\langle W^1,W^2\rangle_t&=\xi(t,Z_t,r_t)\diff t,
\end{align*}
where $a,\sigma_r>0$ are constants and $b(\cdot)$ is a function of time, calibrated so that the dynamics of $(r_t)_{0\leq t\leq T}$ match the market data (e.g., suitable interest rates caps and floors). Both $a$ and $\sigma_r$ being positive constants is not a particularly restrictive constraint as remarked in \textcite{brigo2007interest,hull1995note}. Note that $b(t)$ will therefore need to be calibrated to fit the term structure of interest rates seen in the market. Our aim now is to calibrate $\sigma(t,Z_t,r_t)$ and $\xi(t,Z_t,r_t)$ using our OT-methodology. In order to calibrate the local volatility and correlation, we will want to find a cost function that forces $\alpha_t^{\IP}$ and $\beta_t^{\IP}$ to take the form above. As discussed before, we achieve this by using a functional form for $F$ as long as $(\alpha, \beta)\in \Gamma$ for some convex set $\Gamma$, with $F=+\infty$ otherwise. The set $\Gamma$ will enforce in particular that $\beta$ is positive semidefinite and symmetric, and $\beta_{22}=\sigma_r^2$, which defines a set that is convex in $\beta$. With no further constraints on $\beta_{12}$, we obtain a non-explicit form of the Legendre transform of the cost function, $F^*$, which in general will require a numerical solver at each $(t,x)$. To reduce the computational difficulty of the numerical solution, as in \textcite{guo2022calibration}, we first restrict the correlation to the following form, which still keeps the set $\Gamma$ convex
\begin{equation}\label{eq: convexity adjustment}
\xi(t,Z_t,r_t)=\frac{\sigma_r}{\sigma(t,Z_t,r_t)}\xi_{\mathrm{ref}}(t,Z_t,r_t),\quad\text{for}\:t\in[0,T],
\end{equation}
where $\xi_{\mathrm{ref}}(t,Z_t,r_t)\in\IR^2$ here is a fixed reference function. We then set
\begin{align}
\Gamma(t,Z_t,r_t)=\bigg\{(\alpha,\beta)\in\IR^2\times\IS^2\: :\:&\alpha_1=r_t-\frac{1}{2}\beta_{11},\:\alpha_2=b(t)-ar_t,\notag\\
&\beta_{12}=\beta_{21}=\xi_{\mathrm{ref}}\sigma_r^2,\:\beta_{22}=\sigma_r^2\bigg\},\label{eq: gamma main case}
\end{align}
where one would change the set $\Gamma$ suitably if a different short rate model was fixed initially. 
We also require the inequality $\xi_{\mathrm{ref}}(t,Z_t,r_t)^2\sigma_r^2\leq\sigma^2(t,Z_t,r_t)$ to keep $\xi(t,Z_t,r_t)\in [-1,1]$ as a correlation also. Since $r_t$ is on a much lower scale to $Z_t$, we have that $\sigma_r^2\ll\sigma^2$, so this condition is not financially restrictive. To enforce this condition, we will define a convex function $H:\IR\times\IR^+\times\IR\to\IR\cup\{ +\infty\}$ with a parameter $p>2$
$$
H(x,\bar{x},s)\coloneqq\begin{cases}
(p-1)\big(\frac{x-s}{\bar{x}-s}\big)^{1+p}+(p+1)\big(\frac{x-s}{\bar{x}-s}\big)^{1-p}-2p,&\text{if}\: x,\bar{x}>s,\\
+\infty,&\text{otherwise}.
\end{cases}
$$
The coefficients of each term ensure that $H(x,\bar{x},s)$ is minimised over $x$ at $x=\bar{x}$ with $\min_x H(x,\bar{x},s)=0$. We fix a reference local volatility function $\bar{\sigma}^2=\bar{\sigma}^2(t,Z_t,r_t)$ that represents the desired model. Our aim is to find a calibrated model which does not deviate too much from the reference one. To achieve this we set 
\begin{equation}\label{eq: cost function}
F(t,Z,r,\alpha,\beta)=\begin{cases}
H(\beta_{11},\bar{\sigma}^2,\xi_{\mathrm{ref}}^2\sigma_r^2),&\text{if}\:(\alpha,\beta)\in\Gamma(t,Z_t,r_t),\\
+\infty,&\text{otherwise}.
\end{cases}
\end{equation}
\begin{remark}
It is easy to check that the function $F$ defined in (\ref{eq: cost function}) satisfies the assumptions for the duality proof.
\end{remark}
This cost function will ensure that we retain the Hull-White model in the interest rate, while also matching the market prices for the call options by calibrating the volatility of the stock. Additionally, we wish to enforce that the matrix $\beta$ from our model characteristics is positive definite and that $\xi_{\mathrm{ref}}$ remains a correlation function, and we achieve this by setting $s=\xi_{\mathrm{ref}}^2\sigma_r^2$ as an argument of $H$ in the definition of $F$. Applying Theorem~\ref{thm: dual problem}, we have the following dual formulation with the given cost function $F(\alpha,\beta)$
\begin{problem}
$$
V=\sup_{\lambda}\lambda\cdot u-\phi^{\lambda}(0,Z_0,r_0).
$$
Where $\phi^{\lambda} = \phi(t,z,r)$ solves the HJB equation for $(t,z,r)\in [0,T]\times\IR^2$
\begin{align}
\sum_{i=1}^n\lambda_i(\exp(z)-K_i)^+\delta_{\tau_i}&+\partial_t\phi+\sup_{\beta_{11}}\bigg\{(r-\frac{1}{2}\beta_{11})\partial_{z}\phi+(b(t)-ar)\partial_{r}\phi+\frac{1}{2}\beta_{11}\partial^2_{zz}\phi\notag\\&+\bar{\xi}\sigma_r^2\partial^2_{zr}\phi+\frac{1}{2}\sigma^2_r\partial^2_{rr}\phi-r\phi-H(\beta_{11},\bar{\sigma}^2,\bar{\xi}^2\sigma_r^2)\bigg\}=0.\label{eq: HW HJB}
\end{align}

\end{problem}
\begin{lemma}\label{lem: sup beta}
The optimal characteristic in the HJB equation (\ref{eq: HW HJB}), $\beta_{11}^*$ is given by
\begin{align}
\beta_{11}^*(t,z,r) = \bar{\xi}^2\sigma_r^2+&\Biggl(\frac{(\bar{\sigma}^2-\bar{\xi}^2\sigma_r^2)^p(\phi_{zz}-\phi_z)}{4(p^2-1)}\notag\\
&+\frac{1}{2}\sqrt{(\frac{(\bar{\sigma}^2-\bar{\xi}^2\sigma_r^2)^p(\phi_{zz}-\phi_z)}{2(p^2-1)})^2+4(\bar{\sigma}^2-\bar{\xi}^2\sigma_r^2)^{2p}}\Biggr)^{\frac{1}{p}}.\label{eq: optimal beta11}
\end{align}
\end{lemma}
\begin{proof}
By differentiating the argument of the supremum in (\ref{eq: HW HJB}) with respect to $\beta_{11}$, we notice that solving the supremum over $\beta_{11}$ in (\ref{eq: HW HJB}) is equivalent to solving the equation for $x$:
\begin{equation*}
\frac{1}{2}(\partial^2_{zz}\phi-\partial_z\phi)=\partial_x H(x,\bar{\sigma}^2,\bar{\xi}^2\sigma_r^2).
\end{equation*}
Computing the right hand term and rearranging, we have
$$
\frac{\partial}{\partial x}H(x,\bar{\sigma}^2,\bar{\xi}^2\sigma_r^2) = (p^2-1)\bigg[\bigg(\frac{x - \bar{\xi}^2\sigma_r^2}{\bar{\sigma}^2 - \bar{\xi}^2\sigma_r^2}\bigg)^p - \bigg(\frac{x - \bar{\xi}^2\sigma_r^2}{\bar{\sigma}^2 - \bar{\xi}^2\sigma_r^2}\bigg)^{-p}\bigg]
$$
We again rearrange and arrive at the quadratic in $(x - \bar{\xi}^2\sigma_r^2)^p$
$$
(x - \bar{\xi}^2\sigma_r^2)^{2p} - (x - \bar{\xi}^2\sigma_r^2)^p(\bar{\sigma}^2 - \bar{\xi}^2\sigma_r^2)^p\frac{\phi_{zz}-\phi_z}{2(p^2-1)} - (\bar{\sigma}^2 - \bar{\xi}^2\sigma_r^2)^{2p} = 0
$$
Thus solving this and taking the positive root since $x>\bar{\xi}^2\sigma_r^2$ gives us the desired answer.
\end{proof}
Once the optimal $\beta^{*}_{11}$ is obtained the full model dynamics are specified and we can compute the model price $\psi(0,z,r)$ of an instrument with payoff $\tilde G$ and maturity $\tilde \tau\leq T$ by solving the standard pricing PDE for $(t,z,r)\in [0,\tilde \tau)\times\IR^2$

\begin{equation}\label{eq: pricing PDE}
\begin{cases}\begin{aligned}\!\partial_t\psi+(r-\frac{1}{2}\beta_{11}^*)\partial_z\psi&+\big(b(t)-ar\big)\partial_r\psi+\frac{1}{2}\beta_{11}^*\partial_{zz}^2\psi\\&+\overline{\xi}\sigma_r^2\partial^2_{zr}\psi+\sigma_r^2\partial^2_{rr}\psi-r\psi=0,\end{aligned}& \\
\psi(\tilde \tau,z,r)=\tilde{G}(z,r).&\end{cases}
\end{equation}

If we denote $\IP^*$ the measure corresponding to the dynamics of the selected model, then via the Feynman-Kac formula, $\IE^{\IP^*}[e^{-\int_0^{\tilde \tau}r_s\,\diff s}\tilde G(Z_{\tilde \tau},r_{\tilde \tau})\big|Z_0=z,r_0=r]=\psi(0,z,r)$. 

\section{Numerical Method}
In this section, we outline the numerical method used to solve the dual formulation in Theorem~\ref{thm: dual problem}. We will use and adapt the methods presented in \textcite{guo2022calibration} and \textcite{guo2022joint} for ease of implementation. Analogously to those two works, in order to compute the supremum in \ref{thm: dual problem}, we must solve an HJB equation for a given $\lambda$, and then update the $\lambda$ using an optimisation algorithm. To speed up the convergence, we compute the gradients with respect to $\lambda$ of the dual objective function.
\begin{lemma}\label{lem: gradients}
Suppose Problem~\ref{prob: primal} is admissible and define the dual objective function as
\begin{equation}\label{eq: dual objective}
L(\lambda) = \lambda\cdot u - \phi^{\lambda}(0,Z_0,r_0).
\end{equation}
Then the gradients of the dual objective function are given by
\begin{equation}\label{eq: dual gradients}
\partial_{\lambda_i}L(\lambda) = u_i - \IE^{\IP}\bigg[e^{-\int_0^{\tau_i}\!r_s\,\diff s}G_i(Z_{\tau_i})\bigg].
\end{equation}
\end{lemma}
The gradients are obtained via the same method in \textcite[Lemma~4.5]{guo2022calibration}, but with the discounting appearing in the expectation as a result of the $-r\phi$ term in the HJB equation from the Feynman-Kac formula. The interpretation is the same here, that the gradients represent the difference between the model and market prices. We briefly outline the numerical method from \textcite{guo2022calibration} and \textcite{guo2022joint} which can be directly applied here. 
We are first given an initial guess $\lambda$, which will usually be taken to be a zero vector, and then solve (\ref{eq: HW HJB}) to obtain $\phi(0,Z_0,r_0)$. Since the HJB equation (\ref{eq: HW HJB}) has dirac deltas at the times $(\tau_i)_{i=1,\dots,n}$, the solution also has jump discontinuities in time. However, in-between these jump discontinuities, the solution is continuous so can be solved using standard techniques backwards in time, and then the jump discontinuities can be incorporated into a terminal condition at the expiration time. That is, we take the final jump discontinuity as a terminal condition, solve the HJB equation up to the next calibrating option expiry, and then add the jump discontinuity to the solution as a new terminal condition at the next jump discontinuity in the same way as given in Definition~\ref{def: viscosity solution} since we understand the solution in the viscosity sense. Then, given $\phi(0,Z_0,r_0)$, we calculate the objective value (\ref{eq: dual objective}), and apply an optimisation algorithm to update the $\lambda$ to a new guess. As in \textcite{guo2022calibration} and \textcite{guo2022joint}, the L-BFGS algorithm of \textcite{liu1989limited} displayed good convergence, and with the gradients calculated in Lemma~\ref{lem: gradients} the convergence can be made faster. We remark that to compute (\ref{eq: dual gradients}), the $\IE^{\IP}\bigg[e^{-\int_0^{\tau_i}\!r_s\,\diff s}G_i(Z_{\tau_i})\bigg]$ term is simply the model price of the $i^{\mathrm{th}}$ option. This expectation can either be computed via Monte Carlo methods, or using standard pricing PDE techniques by solving (\ref{eq: pricing PDE}) with $\beta_{11}^*$ obtained from the solution of the HJB equation. This process is then repeated until $||\nabla_{\lambda}L(\lambda)||_{\infty}<\epsilon$ for some specified tolerance $\epsilon>0$, which corresponds to the model prices matching all of the market prices.

We perform the constant rescaling in the short rate variable $r_t\mapsto R r_t$ where we choose $R=100$ for stability reasons in the finite difference approximations to make it of the same order as the other coordinate. In addition, we rescale the calibrating option prices and their payoffs by their vegas computed from their Black-Scholes implied volatility. This not only helps the stability of the numerical method by reducing the magnitude of the jump discontinuities, but it also converts pricing errors into implied volatility errors since the vega represents how much the option price will change as the volatility changes by $1\%$. 
\subsection{Numerically Solving the HJB Equation}
We now outline how to solve the HJB equation (\ref{eq: HW HJB}). In \textcite{barles1991convergence}, it has been shown that monotonicity, stability and consistency of a numerical scheme guarantees convergence locally uniformly so long as there exists a comparison principle for the analytic solution. We use a policy iteration method (see \textcite{ma2017unconditionally}) similar to that in \textcite{guo2022calibration} and \textcite{guo2022joint}, where we solve the HJB equation using an implicit finite difference method, with central difference approximations for the spatial derivatives from \textcite{in2010adi}. We choose a boundary far away enough such that the boundary conditions do not have a significant effect on the HJB equation solution, and our boundary conditions are such that the second derivative of $\phi$ does not change with time between each maturity. That is, for a subsequence of the calibrating option maturity times $(\tau_{i_k})_{k=1,\dots,m}$ such that for $k=1,\dots,m$ all $\tau_{i_k}$ are distinct, with $\tau_{i_0}=0$, we set for $t\in (\tau_{i_{k-1}},\tau_{i_k}]$, for each $k=1,\dots,m$, 
\begin{align}
\partial^2_{zz}\phi(t,z,r) &= \partial^2_{zz}\phi(\tau_i,z,r),\quad\text{for }z\in\{z_{\mathrm{min}},z_{\mathrm{max}}\},\:r\in [r_{\mathrm{min}},r_{\mathrm{max}}],\label{eq: HJB BC1}\\
\partial^2_{rr}\phi(t,z,r)&= \partial^2_{rr}\phi(\tau_i,z,r),\quad\text{for }z\in[z_{\mathrm{min}},z_{\mathrm{max}}],\:r\in \{r_{\mathrm{min}},r_{\mathrm{max}}\}.\label{eq: HJB BC2}
\end{align}
We started with this boundary condition and compared it with constant Dirichlet boundary conditions $\phi(t,x)=C_{i_k}$, for some constant $C_{i_k}\in\IR$, all $t\in(\tau_{i_{k-1}},\tau_{i_k}]$, and $x$ on the boundary of our rectangular domain, and observed no significant change. In addition, we use the boundary conditions \eqref{eq: HJB BC1}-\eqref{eq: HJB BC2} in the linearised pricing PDE \eqref{eq: pricing PDE}. At each time step, we start with the value of $\phi$ at the previous time step, we then approximate the PDE coefficients $\alpha$ and $\beta$ using Lemma~\ref{lem: optimisers} and Lemma~\ref{lem: sup beta}, and then solve the HJB equation at that time step using one step of a fully implicit spatially second order finite difference scheme. We then use Lemma~\ref{lem: optimisers} and Lemma~\ref{lem: sup beta} again at the same time step to approximate the PDE coefficients again with the new value of $\phi$. This process is repeated until convergence within some specified tolerance, after which we procede to the next time step. Once we have solved the HJB equation at all time steps, we then use the value of $\beta_{11}$ computed from solving the HJB equation, compute $\beta_{12}$ from the convexity adjustment, and use them to solve the linearised model pricing PDE via the ADI method to generate the model prices. In the simulated data example, we use a discretisation on a uniform $100\times 100$ spatial grid of $[4,5]\times [0,5]$ for the log-stock and rescaled interest rates, and partition the time interval into year fractions at the resolution of one day, so that $\diff t = \frac{1}{365}$. We use the minimisation package, minfunc, of \textcite{schmidt2005minfunc} for the implementation of L-BFGS.

The above numerical method is summarised in Algorithm 1. We let $0=t_0<t_1<\dots<t_N=T$ be a discretisation of the time interval $[0,T]$ such that the expiration times of the calibrating instruments $(\tau_i)_{i=1,\dots,n}$ are a subset of the discretisation times. We let $\epsilon_1$ be the tolerance of the model prices to the calibrating prices, so that $||\nabla_{\lambda} L(\lambda)||_{\infty}<\epsilon_1$, where the gradients are given in Lemma~\ref{lem: gradients}, and we let $\epsilon_2$ be the tolerance of the policy iteration for approximating the optimal characteristics.\par

\begin{algorithm}[!ht]
\caption{Policy iteration algorithm.}
\DontPrintSemicolon
\SetNoFillComment
\KwData{Input an initial $\lambda$ and market prices $u_i$.}
\KwResult{Calibrated model prices, optimal characteristics}
\While{$||\nabla_{\lambda}L(\lambda)||_{\infty}>\epsilon_1$}{
\tcc{Solve the HJB equation backwards in time}
\For{$k=N-1,\dots,0$}{
\tcc{Terminal Conditions - adding $\lambda$ multiplied by the payoff}
\If{$t_{k+1}=\tau_i$ for some $i=1,\dots,n$}{
$\phi_{t_{k+1}}\gets\phi_{t_{k+1}}+\sum_{i=1}^n\lambda_iG_i\indic{\{t_{k+1}=\tau_i\}}$\
}
\tcc{Policy iteration to approximate the optimal characteristics}
$\phi^{\mathrm{new}}_{t_k}\gets\phi_{t_{k+1}}$\tcp*[r]{Approximate using previous time step}
\While{$||\phi^{\mathrm{new}}_{t_k}-\phi^{\mathrm{old}}_{t_k}||>\epsilon_2$}{
$\phi^{\mathrm{old}}_{t_k}\gets\phi^{\mathrm{new}}_{t_k}$\tcp*[r]{Store the old value of $\phi$}\
Approximate $\beta^*_{11}$ using Lemma~\ref{lem: optimisers} and Lemma~\ref{lem: sup beta} with $\phi^{\mathrm{old}}_{t_k}$. \tcp*[r]{Use old values to approximate optimal characteristics}\
Use $\beta^*_{11}$ to compute $\alpha^*_1$ and $\beta^*_{12}$, then plug into (\ref{eq: HW HJB}) to remove the supremum and solve using one step of an implicit finite difference method, and set the solution to $\phi^{\mathrm{new}}_{t_k}$.
}
$\phi_{t_k}\gets\phi^{\mathrm{new}}_{t_k}$\tcp*[r]{Save the solution once the $\phi$ has converged to the optimal solution}
}
\tcc{Computing the model prices and gradients}
Compute the model prices by solving the pricing PDE (\ref{eq: pricing PDE}) using the ADI method.\;
Compute the gradients (\ref{eq: dual gradients}).\;
Use the L-BFGS algorithm to update $\lambda$.
}
\end{algorithm}
\subsection{Numerical Results}\label{sec: numerical results}
We present numerical results showcasing the performance of our proposed calibration method. We use simulated data to investigate the advantages and drawbacks of our method, and in particular its dependence on the reference model $\bar{\sigma}$ in \eqref{eq: cost function}. We use the CEV (constant elasticity of variance) model of \textcite{cox1996constant}, with different sets of parameters for the model used to generate the option prices and for our reference model. The underlying $S_t = \exp(Z_t)$, $0\leq t\leq T$, thus solves the following stochastic differential equation:
\begin{equation*}
\diff S_t = r_t S_t\diff t + \sigma(t,S_t)S_t\diff W^1_t,
\end{equation*}
with $\sigma(t,S_t) = \sigma S_t^{\gamma - 1}$, where $\sigma\geq 0$ and $\gamma\geq 0$ are both constants. We remark that this is a special case of the SABR (``stochastic $\alpha$, $\beta$, $\rho$'') model derived in \textcite{hagan2002managing} as a stochastic volatility extension of the CEV model.

We solve a pricing PDE to compute the generating model prices and consider the following instruments:
\begin{enumerate}
\item calls on the underlying with an expiration of 60 days;
\item calls on the underlying with an expiration of 120 days.
\end{enumerate}

We note that the payoffs of these options are not smooth and cause instabilities when computing the derivatives in the terminal conditions of the HJB equation. Thus, we use the smoothed version of the call option payoffs, which for $\epsilon\ll K$ is given by:
\begin{equation*}
(S_T-K)^+\approx\frac{S_T-K}{2}\Big[\mathrm{tanh}\Big(\frac{S_T-K}{\epsilon}\Big)+1\Big].
\end{equation*}
We keep the interest rate model parameters the same throughout since that is assumed to be given. 

We present two numerical examples, the ``good'' reference model where the initial reference model is parametrically close to the generating model, and the ``bad'' reference model where it is not parametrically close to the generating model and in particular has a correlation with a different sign. Note that in the extreme case when the generating model and the reference model are the same, the calibration procedure will stop instantly and recover the generating model. The parameters for all the models are summarised in Table \ref{table: HW CEV parameters}.
\begin{table}[htp]
\centering
\begin{tabular}{|l|l|p{0.55\linewidth}|}
\hline
Parameter & Value & Interpretation\\
\hline
$Z_0$ & $\log(92)$ & Initial log-underlying price\\
$r_0$ & $0.025\times 100$ & Initial short rate scaled by $R=100$\\
$\epsilon_1$ & $1\times 10^{-4}$ & Tolerance for the difference in model and market implied volatility\\
$\epsilon_2$ & $1\times 10^{-8}$ & Tolerance for the policy iteration approximation of the optimal characteristics\\
$p$ & 4 & Exponent in the cost function\\
\hline
$\sigma$ & 0.78 & Volatility scaling of generating CEV model\\
$\gamma$ & 0.9 & Power law in generating CEV model\\
$\theta(t)$ & $ar_0+\frac{\sigma_r^2}{2a}(1-e^{-2at})$ & Initial term structure of Hull-White generating model\\
$a$& 0.4 & Speed of mean reversion of Hull-White generating model\\
$\sigma_r$ & 0.05 & Volatility of Hull-White generating model\\
$\xi$ & $-0.6$  & Instantaneous correlation between short rate and log-stock in generating model\\
\hline
$\overline{\sigma}_{\mathrm{good}}$ & 0.9 & Volatility scaling of the ``good'' reference CEV model\\
$\overline{\gamma}_{\mathrm{good}}$ & 0.9 & Power law in the ``good'' reference CEV model\\
$\overline{\xi}_{\mathrm{good}}$ & $-0.4$ & Instantaneous correlation between short rate and log-stock in the ``good'' reference model\\
\hline
$\overline{\sigma}_{\mathrm{bad}}$ & 1.2 & Volatility scaling of the ``bad'' reference CEV model\\
$\overline{\gamma}_{\mathrm{bad}}$ & 0.78 & Power law in the ``bad'' reference CEV model\\
$\overline{\xi}_{\mathrm{bad}}$ & 0.4 & Instantaneous correlation between short rate and log-stock in the ``bad'' reference model\\
\hline
\end{tabular}
\caption{Parameter values for the simulated data example.}
\label{table: HW CEV parameters}
\end{table}
Our numerical methods converged for both the ``good'' and ``bad'' reference model case, calibrating all the call options to a tolerance of $10^{-4}$, with the calibrated model implied volatility replicating the prices of the generating model. The generating and calibrated model implied volatility skews are indistinguishable for both maturities within the range $K=[85,120]$ in which our options' strikes were taken, indicating that our optimal transport model replicates the observed simulated data both at and in-between the calibrating option strikes. The dependence on the reference model is only observed outside of this range, which is to be expected. The plots are given in Figure \ref{fig:SPX IVOL} with the exact results summarised in Table~\ref{table: simulated data}. 

A feature of our method is that it is designed to find a calibrated model closest to the reference model, in the sense of minimising \eqref{eq: objective function}. This results in a spiky volatility surface -- the method tries to stick to the reference model whilst making sharp deviations needed to calibrate to the given options' prices. This is not necessarily a desired feature and we propose to smooth the surface to obtain a reasonable volatility surface while matching the market data. We used an iterative procedure:  once we obtained convergence of the model prices to the generating model prices within the tolerance of $\epsilon_1$, we then stored those characteristics as the new reference model, applied an interpolation method using cubic splines to smooth the surface and then restarted the process with the smoothed surface as our reference model.  We remark that the final epoch of the calibration algorithm involved no interpolation, as this spoils calibration, so the overall algorithm provided the OT calibrated surface after iterating through smoothed reference models. In addition to providing more reasonable model dynamics, the smoothed reference model iteration also improved the numerical stability of the algorithm, and thus allowed us to achieve a better calibration. This came at a significant computational cost due to the iterations of the calibration routine, however both ran on a standard notebook laptop in a matter of hours. Both reference models were smoothed around 10 times before changes in the volatility surface plots were no longer noticeable.

We remark that the aforementioned reference model iteration could, in principle, be circumvented via a regularisation technique by encoding a smoothing penalty into the cost function, such as via a second order Tikhonov regularisation method. Clearly this would change the optimal characteristics given in (\ref{eq: optimal beta11}), but the optimal surface would then have this smoothing penalty encoded into it. However, this would require at all gridpoints to estimate the first and second derivatives of $\beta_{11}$ via a central difference method, which would substantially increase the computational load beyond that of the reference model iteration method, and additionally potentially require extra state variables to account for the spatial derivatives of $\beta$.

\begin{figure}[htp]
\centering
	\subfigure[]{\includegraphics[width=0.48\textwidth]{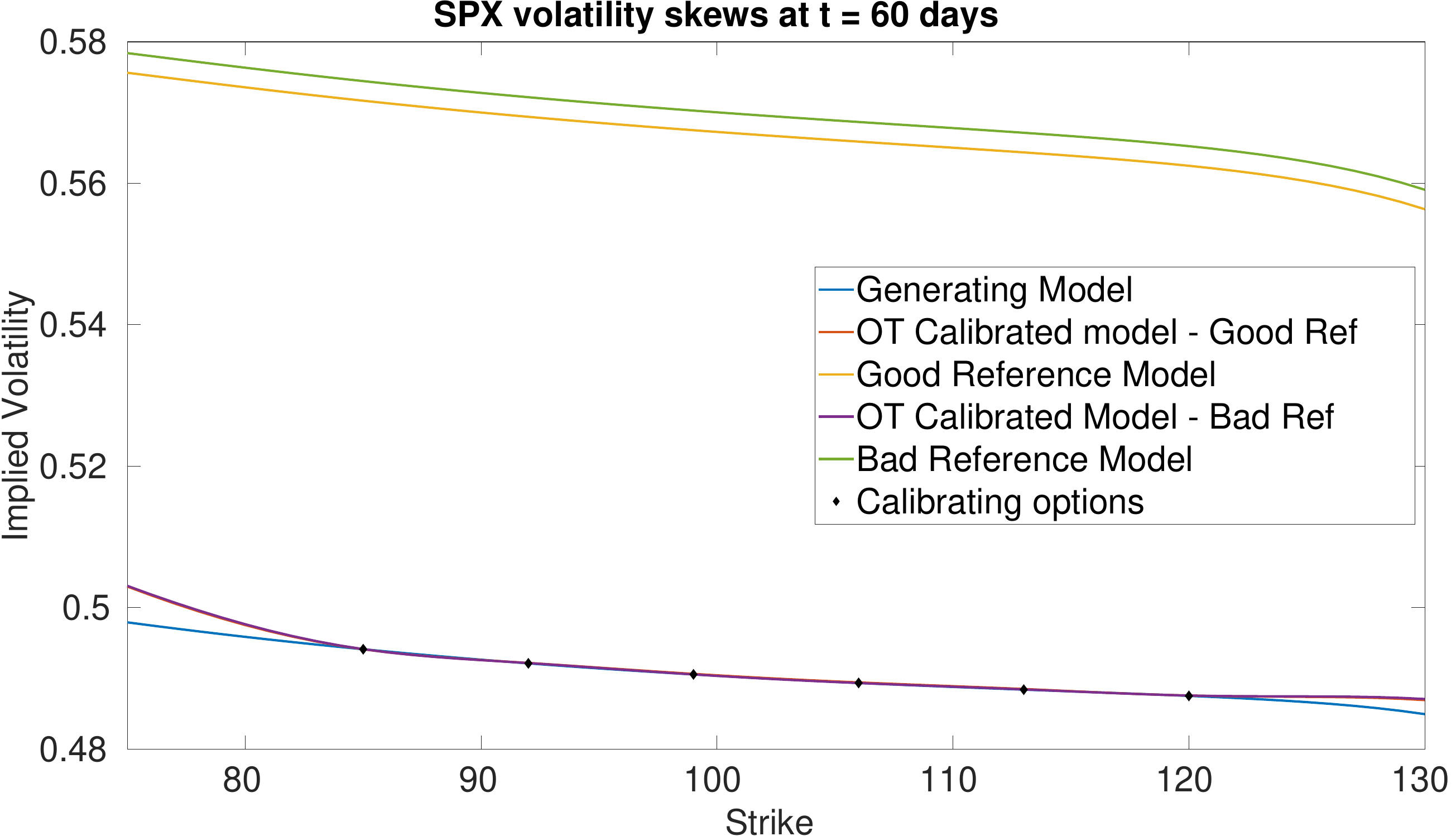} }
	\subfigure[]{\includegraphics[width=0.48\textwidth]{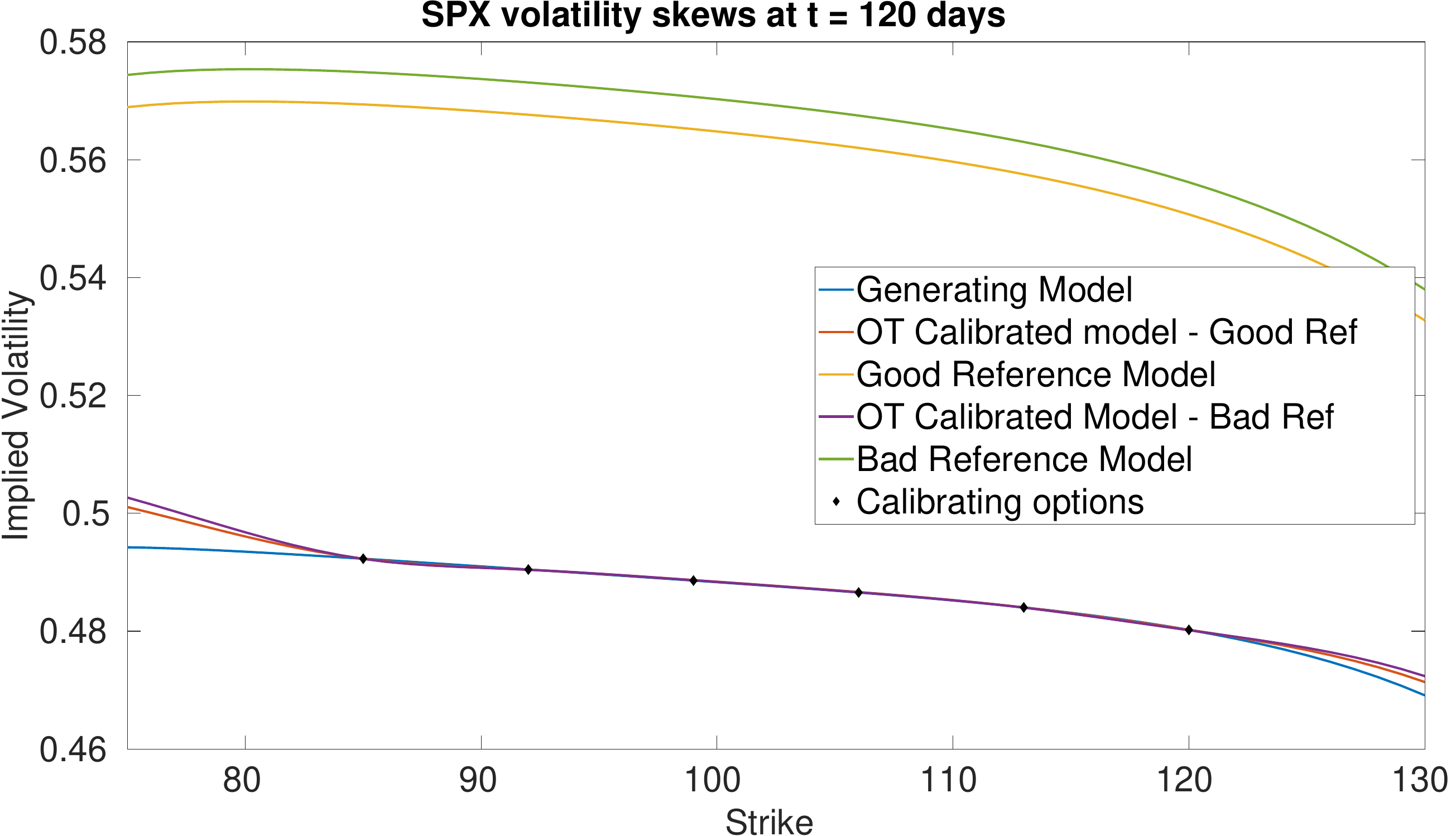} }
	\caption{Implied volatility skews under the generating model, reference model and OT-calibrated models with both reference models and across two maturities.}
	\label{fig:SPX IVOL}
\end{figure}


\begin{table}[htp]
\centering
\scalebox{0.95}{
\begin{tabular}{|c|c|c|c|c|c|c|c|}
\hline
 & &\multicolumn{2}{c|}{Generating Model} & \multicolumn{2}{c|}{Calibrated Model:} & \multicolumn{2}{c|}{Calibrated Model:}\\
& & \multicolumn{2}{c|}{} &\multicolumn{2}{c|}{Good Reference} & \multicolumn{2}{c|}{Bad Reference}\\
\hline
Option Type&Strike&Price&IV&Price&IV&Price&IV\\
\hline
\multirow{6}{*}{\shortstack[1]{SPX Call options\\at $t=60$ days}} & 85 & 11.3666 & 0.4941 & 11.3666 & 0.4941 & 11.3668 & 0.4941\\
& 92 & 7.5389 & 0.4921 & 7.5398 & 0.4922 & 7.5396 & 0.4922\\
& 99 & 4.7538 & 0.4906 & 4.7549 & 0.4906 & 4.7537 & 0.4905\\
& 106 & 2.8616 & 0.4893 & 2.8625 & 0.4894 & 2.8613 & 0.4893\\
& 113 & 1.6523 & 0.4884 & 1.6532 & 0.4885 & 1.6526 & 0.4884\\
& 120 & 0.9189 & 0.4875 & 0.9192 & 0.4876 & 0.9192 & 0.4876\\
\hline
\multirow{6}{*}{\shortstack[1]{SPX Call options\\at $t=120$ days}} & 85 & 14.2787 & 0.4923 & 14.2787 & 0.4923 & 14.2780 & 0.4923\\
& 92 & 10.7017 & 0.4905 & 10.7007 & 0.4904 & 10.7009 & 0.4904\\
& 99 & 7.8563 & 0.4886 & 7.8580 & 0.4887 & 7.8575 & 0.4886\\
& 106 & 5.6560 & 0.4866 & 5.6575 & 0.4866 & 5.6568 & 0.4866\\
& 113 & 3.9917 & 0.4840 & 3.9918 & 0.4840 & 3.9910 & 0.4840\\
& 120 & 2.7493 & 0.4802 & 2.7495 & 0.4802 & 2.7483 & 0.4802\\
\hline
\end{tabular}}
\caption{Table of the generating and calibrated model prices and implied volatilities.}\label{table: simulated data}
\end{table}

To better illustrate the features of the OT-calibrated model, we present plots of the surfaces of the model characteristics, both on their own -- see Figures \ref{fig: beta11} and \ref{fig: xi} -- and superimposed with the characteristics of the generating and reference models for comparison, see Figures \ref{fig: beta11 together} and \ref{fig: xi together}. Broadly speaking, as discussed above, the OT-calibrated model is close to the generating one in the region specified by the data and close to the reference one otherwise. 
In addition, the surfaces we obtain for the correlation $\xi$ are highly dependent on the reference model, as one may expect from the convexity adjustment in (\ref{eq: convexity adjustment}). 
The choice of reference value will therefore mean that $\xi$ is always close to the reference model in this case with some fluctuations obtained through the change in the value $\beta_{11}$.
Consequently, the modeller's (or trader's) insight in specifying a correct correlation (not least its sign) is, relatively speaking, more important as the vanilla option prices' market data alone will not necessarily help to correct a widely wrong reference guess. 
\begin{figure}[H]
\centering
\includegraphics[width=0.85\textwidth]{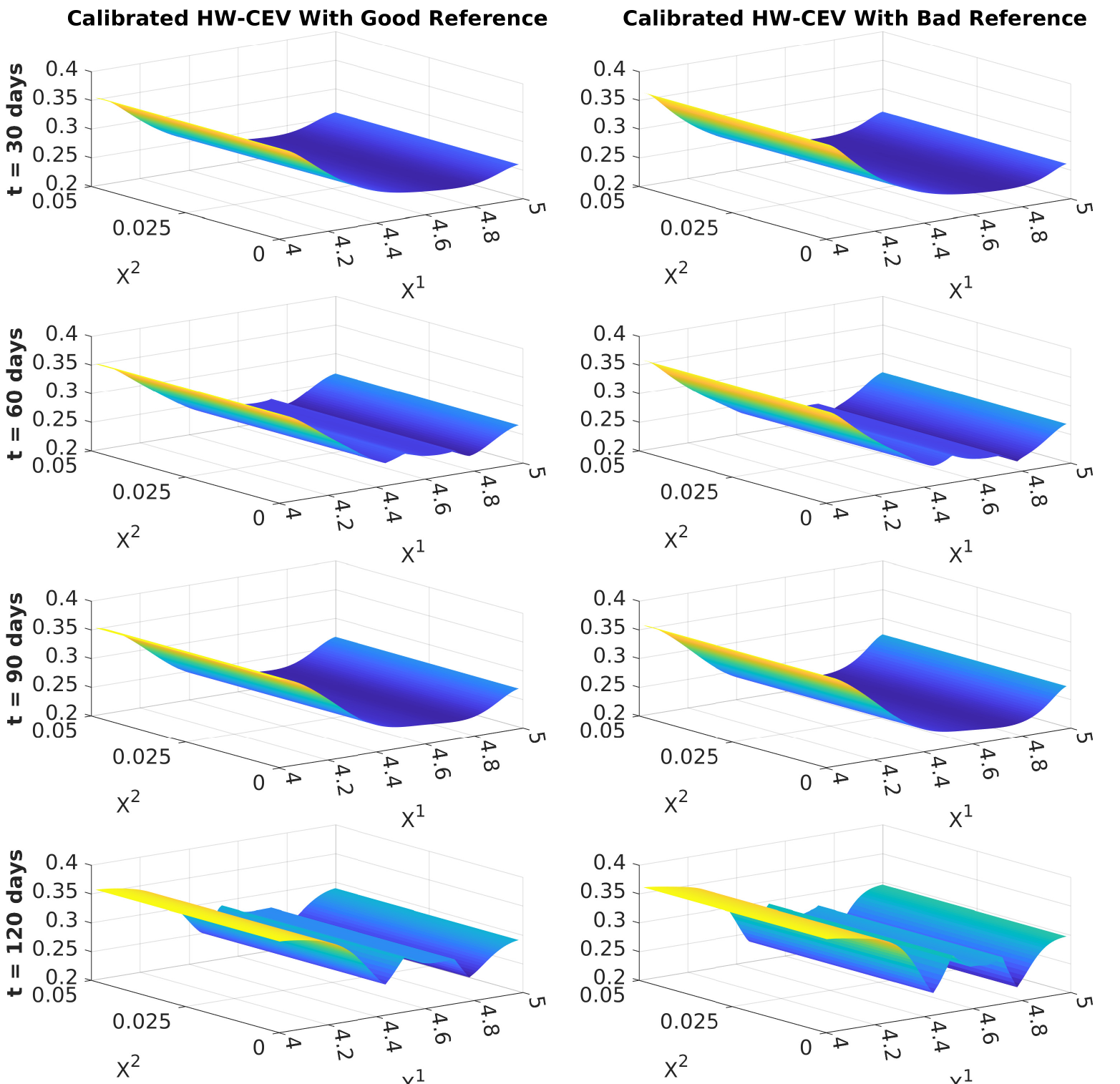}
\caption{Plots of $\beta_{11}$ at $t=30$, 60, 90, and 120 days for the calibrated model.}
\label{fig: beta11}
\end{figure}
\begin{figure}[H]
\centering
\includegraphics[width=0.85\textwidth]{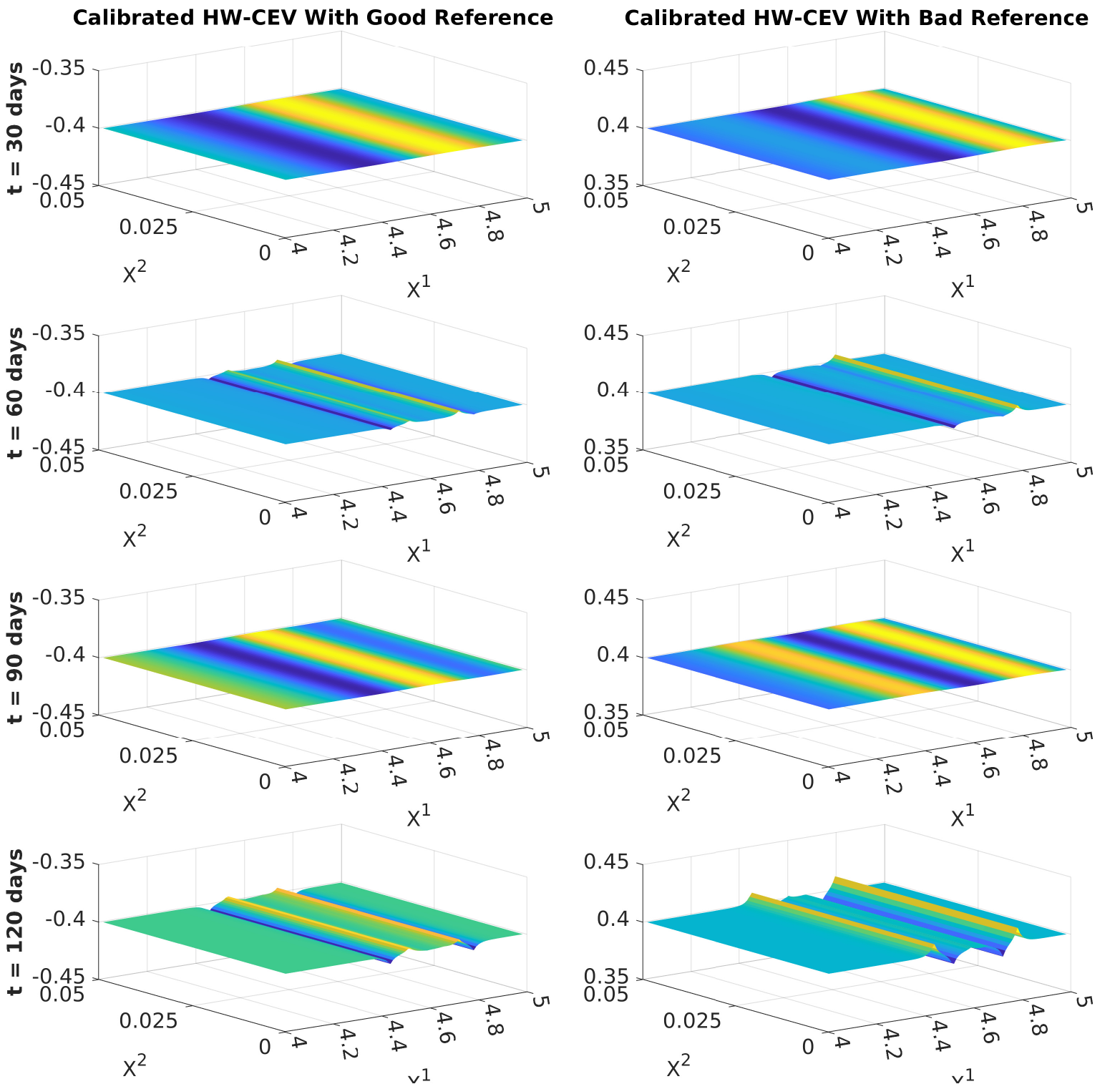}
\caption{Plots of $\xi$ at $t=30$, 60, 90, and 120 days for the calibrated model.}
\label{fig: xi}
\end{figure}

\begin{figure}[H]
\centering
\includegraphics[width=0.85\textwidth]{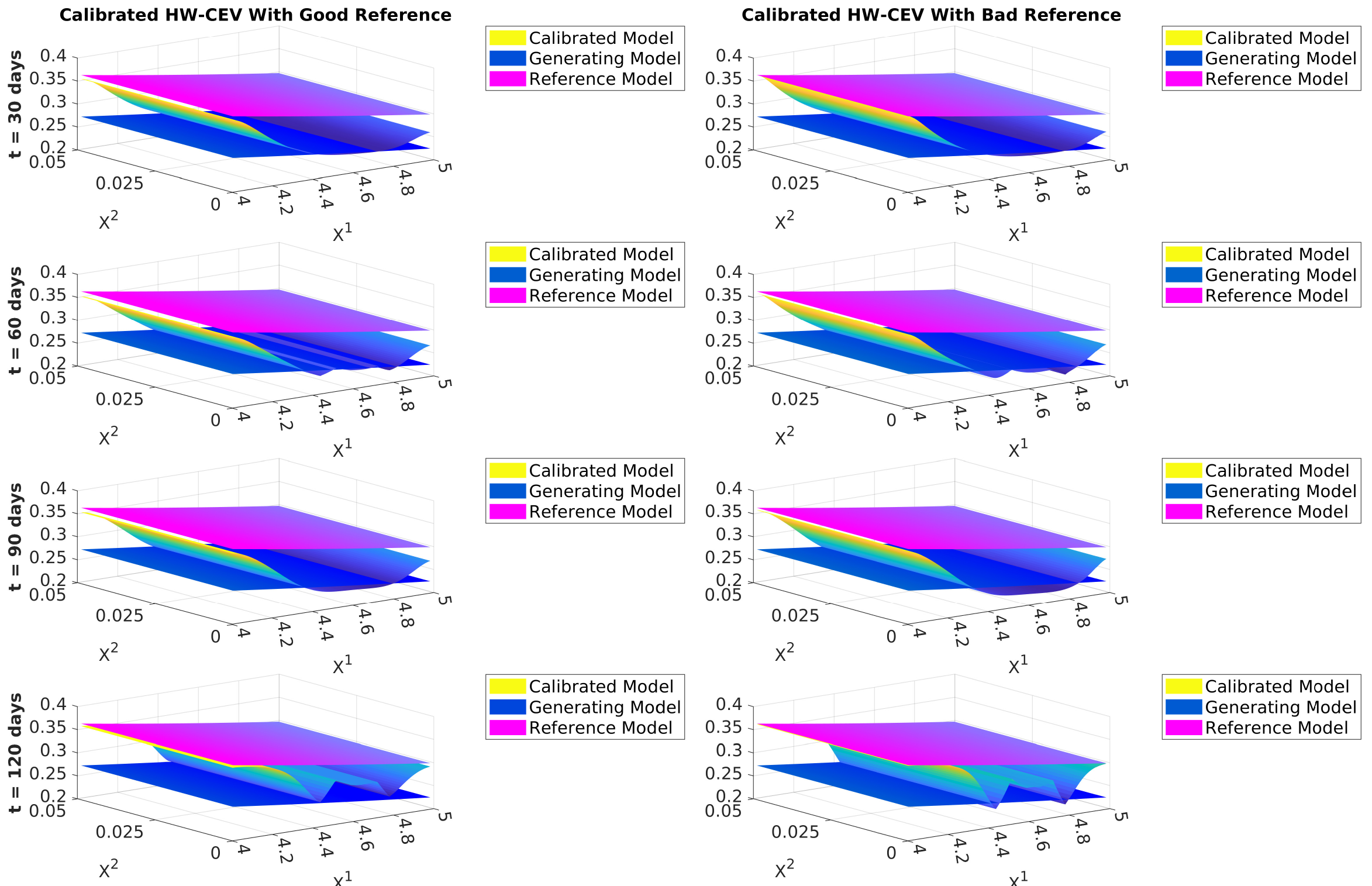}
\caption{Plots of $\beta_{11}$ at $t=30$, 60, 90, and 120 days for the calibrated model compared with the generating model.}
\label{fig: beta11 together}
\end{figure}
\begin{figure}[H]
\centering
\includegraphics[width=0.85\textwidth]{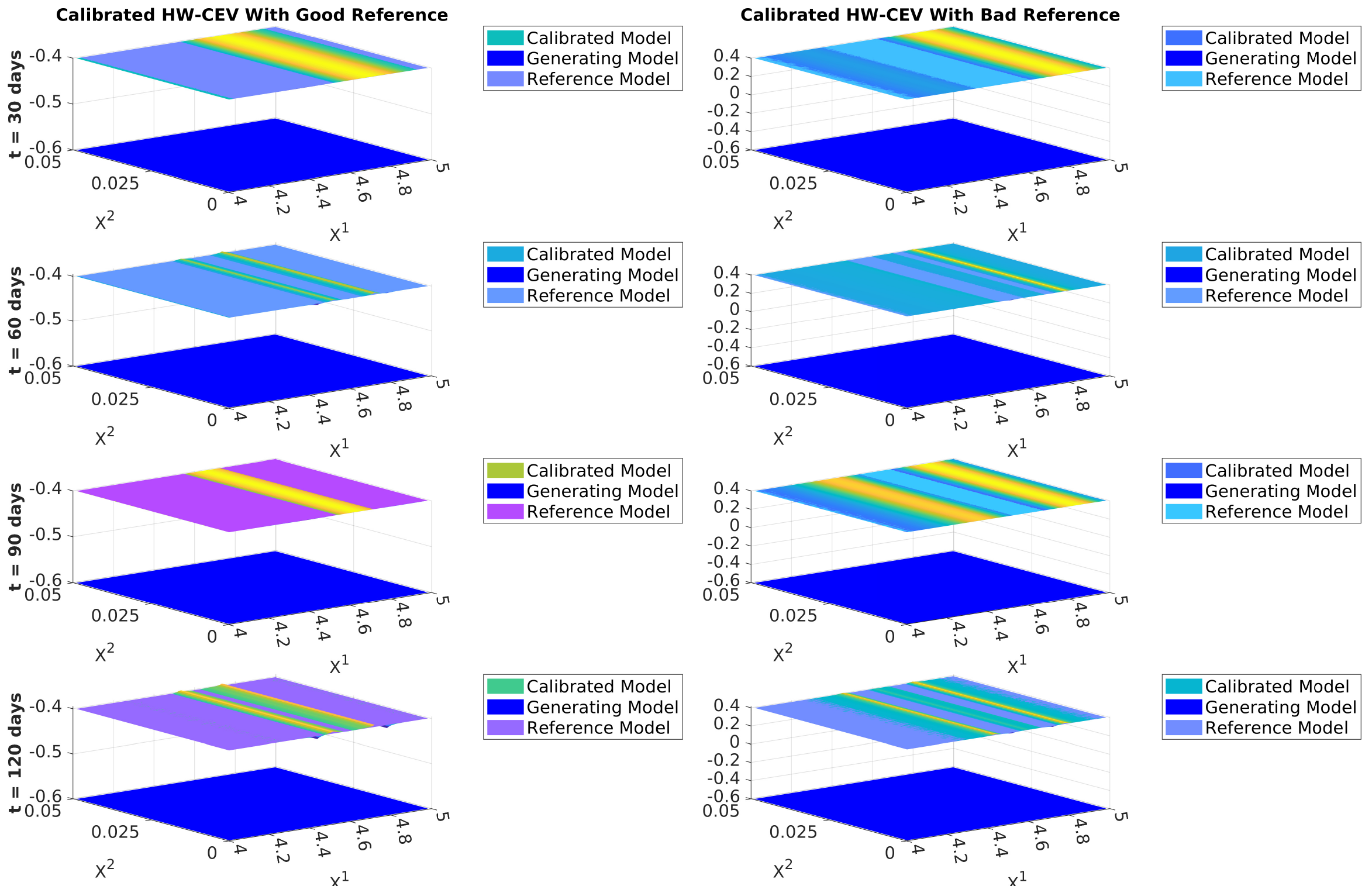}
\caption{Plots of $\xi$ at $t=30$, 60, 90, and 120 days for the calibrated model compared with the generating model.}
\label{fig: xi together}
\end{figure}
\section{Conclusions}
This paper developed a non-parametric optimal-transport driven method to calibrate stock price model in a stochastic interest rate environment. By switching to sub-probability measures representing discounted densities the method manages not to increase the dimension of the state variables. The method is flexible and can accommodate different calibrating instruments. We developed a generic duality result and applied it in the context of Hull-White short rate model and a local volatility stock price model with short rate dependence. We illustrated the numerical performance of the method considering a sequential calibration problem: the short rate model is calibrated first in a parametric setting and it then feeds into our non-parametric local volatility calibration. We highlighted how some model features are pinned down well via the market prices of options. However others, notably the correlation between the Brownian motions driving the rates and the stock price, are much more sensitive to the modeller's choice of their reference model. In a follow up paper \textcite{joseph2023joint}, we apply the duality result developed here to consider simultaneous OT-driven calibration for a joint model for the stochastic interest rates and the stock price process, and demonstrate its performance on real market data. We note that here and in \textcite{joseph2023joint}, we restrict ourselves to short-rate models to keep the dimension of the HJB equation to two state variables. It would be natural and interesting to extend this setup to include more stochastic factors, e.g., stochastic volatility for the stock prices. Each extra factor would raise the dimension of the HJB and imply significant numerical challenge. 
It would also be of interest to consider more realistic fixed income models, such as market models of \textcite{brace1997market,miltersen1997closed,jamshidian1997libor} or more recent multi-curve framework models, see \textcite{henrard2007irony,henrard2010irony}, where the post 2008 financial crisis effects of the LIBOR-OIS spreads are taken into account. This however would drastically increase the dimension of our state process. It would likely necessitate novel tools to solve the resulting PDEs, such as recent machine learning methods, see \textcite{han2018solving,weinan2021algorithms}. We believe these are very interesting avenues for future research. 
\section*{Competing Interests}
The authors declare no competing interests.

\appendix \normalsize
\section{Proof of the Discounted Superposition Principle}
First, we define two operators associated with the discounted Fokker-Planck equation \eqref{eq: discount fp} and the augmented Fokker-Planck equation \eqref{eq: MP fokker planck}. Given $a:[0,T]\times\IR^d\to\IR^d$, $b:[0,T]\times\IR^d\to\IS^d_+$, $c:[0,T]\times\IR^d\to\IR$, define
\begin{align*}
C([0,T];\IR^d)\ni f&\mapsto \mathcal{L}^Df=a\cdot\nabla_xf+\frac{1}{2}b:\nabla^2_xf-cf,\\
C([0,T];\IR^{d+1})\ni f&\mapsto \mathcal{L}^Af=a\cdot\nabla_xf+\frac{1}{2}b:\nabla^2_xf+yc\partial_yf.
\end{align*}
Note that if $(\nu_t)_{t\in[0,T]}\subset\cP(\IR^{d+1})$ solves the augmented Fokker-Planck equation associated with $\mathcal{L}^A$, denoted $\mathrm{aFP}(a,b,c)$, then $(\nu_t^D)_{t\in[0,T]}\subset\cM(\IR^d)$ with $\nu_t^D=\int_{\IR}y\nu_t(\cdot,\diff y)$ solves the discounted Fokker-Planck equation associated with $\mathcal{L}^D$, denoted $\mathrm{dFP}(a,b,c)$. This is easily seen since the coefficients $a,b,c$ have no $y$-dependence. 

We construct a sequence of smooth solutions to \eqref{eq: discount fp}, $(\nu^n_t)_{t\in[0,T]}\subset\cM(\IR^d)$, with coefficients $a^n\in C^{\infty}([0,T]\times\IR^d;\IR^d)$, $b^n\in C^{\infty}([0,T]\times\IR^d;\IS^d_+)$, and $c^n\in C^{\infty}([0,T]\times\IR^d;\IR)$ for which we can easily show Theorem~\ref{thm: discounted superposition}. We therefore can construct a sequence of diffusion operators $\mathcal{L}^{A,n}$ corresponding to $\mathrm{aMP}(a^n,b^n,c^n)$ which has solution $\boldsymbol{\eta}^n$ such that $\int_{\IR}y\eta^n_t(\cdot,\diff y)=\nu^n_t$. The tightness and convergence of $\boldsymbol{\eta}^n$ is then shown, and then we use the arguments to generalise the coefficients $a,b,c$.
\subsection{Smooth Approximation of the Augmented Martingale Problem}\label{sec: smooth approx}
\begin{lemma}\label{lem: smooth superposition}
Let $(\nu_t)_{t\in[0,T]}\subset\cM(\IR^d)$ be a smooth solution of $\mathrm{dFP}(a,b,c)$ with coefficients $a\in C^{\infty}([0,T]\times\IR^d;\IR^d)$, $b\in C^{\infty}([0,T]\times\IR^d;\IS^d_+)$, $c\in C^{\infty}([0,T]\times\IR^d;\IR)$. Then there exists a solution $\boldsymbol{\eta}\in\cP(C([0,T];\IR^{d+1}))$ to $\mathrm{aMP}(a,b,c)$ such that $\int_{\IR}y\eta_t(\cdot,\diff y)=\nu_t(\cdot)$.
\end{lemma}
\begin{proof}
Let $\IP\in\cP_1$ and define the process
\begin{equation}\label{eq: localised SDE2}
\begin{cases}
\diff \tilde{X}_t=a_t(\tilde{X}_t)\diff t + [b_t(\tilde{X}_t)]^{\frac{1}{2}}\diff W^{\tilde{\IP}}_t,& 0\leq t\leq T,\\
\tilde{X}_0=x_0.&
\end{cases}
\end{equation}
Since $a$ and $b$ are taken to be smooth, we automatically have existence and uniqueness of \eqref{eq: localised SDE2} (see \textcite[Corollary~6.3.3]{stroock1979multidimensional}). Define $\tilde{Y}_t=\exp(-\int_0^t c_s(\tilde{X}_s)\diff s)$. Then, by the same reasoning as Theorem~\ref{lem: Markovian projection}, we have that the curve of marginal laws of the augmented process $(\tilde{X},\tilde{Y})$ at $t$, given by $\eta_t=(e_t)_{\#}\boldsymbol{\eta}$ solves the Fokker-Planck equation for $(t,x,y)\in[0,T]\times\IR^{d+1}$
\begin{equation*}
\partial_t\eta_t(x,y) + \nabla_x\cdot[\eta_t(x,y)a_t(x)] -\frac{1}{2}\nabla^2_x:[\eta_t(x,y)b_t(x)]-\partial_y[yc_t(x)\eta_t(x,y)]=0.
\end{equation*}
We now define the discounted version of $\eta_t$ by $\eta_t^D(x)\coloneqq\int_{\IR}y\eta_t(x,y)\,\diff y$, then by a similar calculation to that in the proof of Lemma~\ref{prop:discountedFP}, we have that $(\eta_t^D)_{t\in [0,T]}$ solves the discounted Fokker-Planck equation \eqref{eq: discount fp}. By a similar argument to \textcite[Proposition~4.1]{figalli2008existence}, we have uniqueness of \eqref{eq: discount fp}, so $\eta^D_t=\nu_t$. By It\^{o}'s formula, it is immediate that $\boldsymbol{\eta}\in\cP(C([0,T];\IR^{d+1}))$ solves $\mathrm{aMP}(a,b,c)$, so the result follows.
\end{proof}
We let $(\nu_t)_{t\in[0,T]}\subset\cM(\IR^d)$ be a narrowly continuous solution of $\mathrm{dFP}(a,b,c)$ for $a\in L^1(\diff\nu_t\diff t;\IR^d)$, $b\in L^1(\diff\nu_t\diff t;\IS^d)$, and $c\in L^1(\diff\nu_t\diff t;\IR)$. We first build a smooth sequence $\nu_t^n$ approximating $\nu_t$ for all $t\in[0,T]$, associated with diffusion operators $\mathcal{L}^{D,n}$. Then by Lemma~\ref{lem: smooth superposition}, there exists a sequence of solutions to the augmented martingale problem, $\boldsymbol{\eta}^n$, associated with $\mathcal{L}^{A,n}$, such that $\int_{\IR}y\eta_t^n(\cdot,\diff y)=\nu_t^n$. As in \textcite[Section~A.1]{trevisan2016well}, we construct the approximations in two different ways for the proof of the general case.

\noindent
\textbf{Pushforward via Smooth Maps:}

Let $\pi=(\pi^1,\dots,\pi^d)\in C^2(\IR^{d};\IR^{d})$ with uniformly bounded first and second derivatives. Define the standard diffusion operator, $\mathcal{L}^S$ by $C([0,T];\IR^d)\ni f\mapsto\mathcal{L}^S=a\cdot\nabla_xf+\frac{1}{2}b:\nabla^2_xf$, then we can define the approximations of $(a,b,c)$ by
\begin{equation*}
\pi(a)^i_t=\frac{\diff\pi_{\#}[\mathcal{L}^S(\pi^i)\nu_t]}{\diff\pi_{\#}\nu_t},\quad\pi(b)_t^{i,j}\coloneqq\frac{\diff\pi_{\#}[b_t:(\nabla\pi^i\otimes\nabla\pi^j)\nu_t]}{\diff\pi_{\#}\nu_t},\quad\pi(c)_t\coloneqq\frac{\diff\pi_{\#}[c_t\nu_t]}{\diff\pi_{\#}\nu_t}.
\end{equation*}
Note that $\pi$ preserves uniform bounds on $a,b,c$. Additionally, by an application of the chain rule to compute $\mathcal{L}^D(f\circ\pi)$, it is easy to see that $(\pi_{\#}\nu_t)_{t\in[0,T]}\subset\cM(\IR^d)$ solves $\mathrm{dFP}(\pi(a),\pi(b),\pi(c))$, and therefore by Lemma~\ref{lem: smooth superposition} there exists $\boldsymbol{\eta}$ a solution to $\mathrm{aMP}(\pi(a),\pi(b),\pi(c))$ such that $\pi_{\#}\nu_t = \int_{\IR}y\eta_t(\cdot,\diff y)$.

\noindent
\textbf{Mollification by Convolutions:}

Let $\kappa\geq 0$ be a smooth probability density on $\IR^{d}$. Define
\begin{equation*}
a^{\kappa}_t\coloneqq\frac{\diff [(a_t\nu_t)*\kappa]}{\diff[\nu_t*\kappa]},\qquad b^{\kappa}_t\coloneqq\frac{\diff [(b_t\nu_t)*\kappa]}{\diff[\nu_t*\kappa]},\qquad c_t^{\kappa}\coloneqq\frac{\diff [(c_t\nu_t)*\kappa]}{\diff [\nu_t*\kappa]}
\end{equation*}
By \textcite[Lemma~A.1]{trevisan2016well} we have that $a^{\kappa}$, $b^{\kappa}$, and $c^{\kappa}$ preserve integrability, and have uniformly bounded first and second spatial derivatives on compact sets. Therefore, $(\nu_t*\kappa)_{t\in[0,T]}\subset\cM(\IR^d)$ solves $\mathrm{dFP}(a^{\kappa},b^{\kappa},c^{\kappa})$. Additionally, by Lemma~\ref{lem: smooth superposition} there exists $\boldsymbol{\eta}$ a solution to $\mathrm{aMP}(a^{\kappa},b^{\kappa},c^{\kappa})$ such that $\nu_t*\kappa = \int_{\IR}y\eta_t(\cdot,\diff y)$.

\subsection{Compactness and Convergence of the Augmented Martingale Problem}\label{sec: superposition convergence}
We now assume that $\boldsymbol{\eta}^n$ is a sequence of superposition solutions corresponding to an approximation $\nu^n_t=\pi^n_{\#}\nu_t$ in the pushforward case or $\nu_t^n=\nu_t*\kappa^n$ in the mollification case. In the mollification case, we assume that $\kappa^n\to\delta_0$ narrowly as $n\to\infty$. In the pushforward case we assume that $\pi_n\to\mathrm{Id}$ locally uniformly, and $\nabla\pi^n\to\mathrm{Id}$, and $\nabla^2\pi^n\to 0$ pointwise. Assume also that the sequences of derivatives are uniformly bounded so there exists $C\geq 0$ such that $|\nabla\pi^n|,|\nabla^2\pi^n|\leq C$. Compactness for solutions of the augmented martingale problem then follows directly from \textcite[Section~A.2]{trevisan2016well} since we could recast the superposition principle in the setting of \textcite[Theorem~2.5]{trevisan2016well} by considering superposition solutions to $\mathrm{aFP}(a,b,c)$ and then the tightness arguments are directly applicable. 

We are looking for solutions of $\mathrm{aMP}(a,b,c)$ with a discounted density coinciding with a given solution of $\mathrm{dFP}(a,b,c)$. We therefore consider the class of test functions $C^{1,2}_b([0,T]\times\IR^{d+1};\IR)\ni\tilde{f}(x,y)=yf(x)$ where $f\in C^{1,2}_b([0,T]\times\IR^d;\IR)$. We therefore note that we always have
\begin{equation*}
\int\int_s^t(\mathcal{L}_r^A\tilde{f})\circ e_r\diff r\diff\boldsymbol{\eta}=\int_s^t\int_{\IR^d}(\mathcal{L}^D_rf_r)\int_{\IR}y\eta_r(\diff x,\diff y)\diff r.
\end{equation*}
From now, we use the shorthand $\diff\nu = \diff\nu_t\diff t$ for notational ease. Now, let $\overline{\mathcal{L}^D}$ be a discounted diffusion operator with coefficients $\overline{a},\overline{b},\overline{c}$ continuous and compactly supported. In view of the arguments of \textcite[Section~A.3]{trevisan2016well}, due to narrow convergence of $\boldsymbol{\eta}^n$, we only need to provide bounds on 
\begin{equation}\label{eq: discount superposition bounds}
\limsup_{n\to\infty}\int|\mathcal{L}^{D,n}f-\overline{\mathcal{L}^D}f|\diff\nu^n + \int|\mathcal{L}^Df-\overline{\mathcal{L}^D}f|\diff\nu.
\end{equation}

\noindent
\textbf{Pushforward via Smooth Maps:}

Write $\pi(\mathcal{L}^D)$ for the discounted diffusion operator associated with coefficients $\big(\pi(a),\pi(b),\pi(c)\big)$. Then for a fixed $n$, we have
\begin{align}
\int |\pi^n(\mathcal{L}^{D,n})f-\overline{\mathcal{L}^D}f|\diff\pi^n_{\#}\nu&\leq\int |\mathcal{L}^{D,n}(f\circ\pi^n)-(\overline{\mathcal{L}^D}f)\circ\pi^n|\diff\nu,\notag\\
&\leq\int\frac{1}{2}\sum_{i,j=1}^d|a:(\nabla\pi_i\otimes\nabla\pi_j)-\overline{a}^{i,j}\circ\pi|\diff\nu\notag\\
&\quad+\int\sum_{i=1}^d|\mathcal{L}^S(\pi_i)-\overline{a}^i\circ\pi|\diff\nu + \int |c-\overline{c}\circ\pi|\diff\nu.\notag
\end{align}
Since $\pi^n$ converges to the identity, we have that \eqref{eq: discount superposition bounds} is bounded above by
\begin{equation*}
\int\sum_{i,j=1}^d|a:(\nabla\pi_i\otimes\nabla\pi_j)-\overline{a}^{i,j}\circ\pi|\diff\nu+2\int\sum_{i=1}^d|\mathcal{L}^S(\pi_i)-\overline{a}^i\circ\pi|\diff\nu + 2\int |c-\overline{c}\circ\pi|\diff\nu.
\end{equation*}
Since continuous and compactly supported functions are dense in $L^1(\nu)$, the above can be made arbitrarily small by optimising over $\overline{a},\overline{b},\overline{c}$.

\noindent
\textbf{Mollification by Convolution:}

Let $\overline{\omega}$ be a common bounded and continuous modulus of continuity for $\overline{a},\overline{b},\overline{c}$. Let $\overline{\mathcal{L}^{D,n}}$ be the discounted diffusion operator associated with coefficients
\begin{equation*}
\overline{a}^n\coloneqq\frac{\diff[(\overline{a}\nu)*\kappa^n]}{\diff[\nu * \kappa^n]},\qquad \overline{b}^n\coloneqq\frac{\diff[(\overline{b}\nu) *\kappa^n]}{\diff [\nu*\kappa^n]},\qquad \overline{c}^n\coloneqq\frac{\diff[(\overline{c}\nu)*\kappa^n]}{\diff[\nu*\kappa^n]}.
\end{equation*}
We first show $\lim_{n\to\infty}\int|\overline{\mathcal{L}^{D,n}}f-\overline{\mathcal{L}^D}f|\diff\nu^n=0$. Since $||f||_{C^{1,2}}\leq 1$, we have
\begin{align}
\int|\overline{\mathcal{L}^{D,n}}f-\overline{\mathcal{L}^D}f|\diff\nu^n&\leq\int|\overline{a}^n-\overline{a}|\diff\nu^n+\int|\overline{b}^n-\overline{b}|\diff\nu^n + \int|\overline{c}^n-\overline{c}|\diff\nu^n,\notag\\
&\leq 3\int\overline{\omega}(z)\kappa^n(z)\diff z\to 0.\notag
\end{align}
Therefore, we have
\begin{align*}
\limsup_{n\to\infty}\int|\mathcal{L}^{D,n}f-\overline{\mathcal{L}^D}f|\diff\nu^n&=\limsup_{n\to\infty}\int|\mathcal{L}^{D,n}f-\overline{\mathcal{L}^{D,n}}f|\diff\nu^n\\
&\leq\int|a-\overline{a}|+|b-\overline{b}|+|c-\overline{c}|\diff\nu.
\end{align*}
By the same argument as before, we can make the right hand side arbitrarily small which concludes the proof.
\subsection{Generalisation of Drift and Diffusion Coefficients}
We now use the approximation and convergence described in Section~\ref{sec: smooth approx} and Section~\ref{sec: superposition convergence} to generalise the coefficients $a,b$ and prove Theorem~\ref{thm: discounted superposition}.

\noindent
\subsubsection{Smooth and Bounded Coefficients:}

Let $a:[0,T]\times\IR^d\to\IR^d,$ $b:[0,T]\times\IR^d\to\IS^d_+$, and $c:[0,T]\times\IR^d\to\IR$ be Borel maps such that
\begin{equation}\label{eq: smooth bounded coeffs}
\int_0^T||a_t||_{C^2_b(\IR^d)}+||b_t||_{C^2_b(\IR^d)}+||c_t||_{C^2_b(\IR^d)}\diff t<\infty.
\end{equation}
\begin{lemma}
Let $a,b,c$ be Borel maps satisfying \eqref{eq: smooth bounded coeffs}. Then for every $\overline{\nu}\in\cP(\IR^d)$ there exists $\boldsymbol{\eta}\in\cP\big(C([0,T];\IR^{d+1})\big)$ solving $\mathrm{aMP}(a,b,c)$ such that $\int_{\IR}y\eta_0(\cdot,\diff y)=\overline{\nu}(\cdot)$.
\end{lemma}
The proof is the same as \textcite[Theorem~A.6]{trevisan2016well}, but by considering the augmented process $(X,Y)$ where the regular martingale problem (see \textcite[Definition~2.4]{trevisan2016well}) is given by $\mathrm{aMP}(a,b,c)$.
\begin{lemma}
Let $a:[0,T]\times\IR^d\to\IR^d,$ $b:[0,T]\times\IR^d\to\IS^d_+$, and $c:[0,T]\times\IR^d\to\IR$ be Borel maps such that
\begin{equation*}
\int_0^T||a_t||_{C^2(B)}+||b_t||_{C^2(B)}+||c_t||_{C^2(B)}\diff t<\infty,\quad\text{for every bounded open }B\subset\IR^d
\end{equation*}
Let $(\nu_t)_{t\in[0,T]}\subset\cM(\IR^d)$ be a narrowly continuous solution of $\mathrm{dFP}(a,b,c)$. If $\nu_0\leq 0$, then $\nu_t\leq 0$ for all $t\in[0,T]$. Therefore for $\overline{\nu}\in\cM(\IR^d)$ there exists a unique narrowly continuous solution $\nu$ such that $\nu_0=\overline{\nu}$.
\end{lemma}
The proof is similar to that of \textcite[Theorem~A.7]{trevisan2016well}, so we provide a sketch.
\begin{proof}
Let $g\in C^{\infty}_c([0,T]\times\IR^d)$ be non-negative. Let $\chi_R:\IR^d\to[0,1]$ be a cutoff function and let $a^{\epsilon}_R$, $b^{\epsilon}_R$, $c^{\epsilon}_R$ be a mollification with respect to time and space of $a\chi_R$, $b\chi_R$, and $c\chi_R$, so that $a^{\epsilon}_R$, $b^{\epsilon}_R$, $c^{\epsilon}_R$ satisfy \eqref{eq: smooth bounded coeffs}. Let $\mathcal{L}^{D,\epsilon}_R$ be the diffusion operator associated with $\mathrm{dFP}(a^{\epsilon}_R,b^{\epsilon}_R,c^{\epsilon}_R)$. Let $f^{\epsilon}$ solve 
\begin{equation*}
\partial_tf^{\epsilon}=-\mathcal{L}^{D,\epsilon}_Rf^{\epsilon}+g,\qquad f^{\epsilon}_T=0.
\end{equation*}
Recall the definition of $\mathcal{L}^S$ in Section~\ref{sec: smooth approx} in the construction of the pushforward via smooth maps approximation. Since $\nu$ solves $\mathrm{dFP}(a,b,c)$, $f^{\epsilon}\in C^{1,2}_b(\IR^d)$ and $f^{\epsilon},\nu_0\leq 0$, we have
\begin{align}
0&\geq-\int f^{\epsilon}\chi_R\diff\nu_0,\notag\\
&=\int\chi_R\partial_t f^{\epsilon}+\mathcal{L}^D(f^{\epsilon}\chi_R)\diff\nu,\notag\\
&=\int\chi_R(g-\mathcal{L}^{D,\epsilon}_Rf^{\epsilon} +\mathcal{L}^Sf^{\epsilon})+f^{\epsilon}\mathcal{L}^S\chi_R+cf^{\epsilon}\chi_R+b:\nabla f^{\epsilon}\otimes\nabla\chi_R\diff\nu,\notag\\
&\begin{aligned}\geq\int g\diff\nu - \sup_{t\in[0,T]}||f^{\epsilon}_t||_{C^2_b(\IR^d)}\int&\big[\chi_R(|a^{\epsilon}_R-a|+|b^{\epsilon}_R-b|+|c^{\epsilon}_R-c|)\\&+|\mathcal{L}^S\chi_R|+|b||\nabla\chi_R|)\big]\diff\nu.\end{aligned}\notag
\end{align}
The sup is uniformly bounded in $\epsilon>0$, and as $\epsilon\to 0$ $(a_R^{\epsilon},b_R^{\epsilon},c_R^{\epsilon})\to(a,b,c)$ for $|x|\leq R$. Then letting $R\to\infty$, since $\nabla\chi_R\to 0$ and $\mathcal{L}^{S}\chi_R\to 0$ the right hand side converges to $\int g\diff\nu$.
\end{proof}

\noindent
\subsubsection{Bounded Coefficients:}

This case follows in the same manner as the bounded case in \textcite{trevisan2016well}. We now assume that $a,b,c$ satisfy
\begin{equation}\label{eq: bounded coefficients}
\int_0^T\sup_{x\in\IR^d}|a_t(x)|+\sup_{x\in\IR^d}|b_t(x)|+\sup_{x\in\IR^d}|c_t(x)|\diff t<\infty.
\end{equation}
Using the regularisation kernel $\kappa=\exp(-\sqrt{1+|x|^2})$ and $\kappa^{\epsilon}=\epsilon^{-n}\kappa(x/\epsilon)$, we have with $\nu_t^{\epsilon}=\nu_t*\kappa^{\epsilon}$ that $(\nu_t^{\epsilon})_{t\in[0,T]}$ solves $\mathrm{dFP}(a^{\epsilon},b^{\epsilon},c^{\epsilon})$ with coefficients satisfying \eqref{eq: smooth bounded coeffs}. Thus, existence of $\boldsymbol{\eta}^{\epsilon}\in\cP\big(C([0,T];\IR^{d+1})\big)$ solving $\mathrm{aMP}(a^{\epsilon},b^{\epsilon},c^{\epsilon})$ as a discounted superposition solution to $\nu^{\epsilon}$ follows. Since $\nu^{\epsilon}_0$ is a narrowly convergent sequence, there exists $\theta:\IR\to\IR$ increasing with $\lim_{x\to\infty}\theta(x)=\infty$ such that $\sup_{\epsilon>0}\int\theta(|x|)\diff\nu^{\epsilon}_0\leq 1$. The proof of tightness of $\boldsymbol{\eta}^{\epsilon}$ is then an easy modification of that in \textcite{trevisan2016well}, where the de la Vall\'{e}e Poussin criterion modifies \eqref{eq: bounded coefficients} to
\begin{equation*}
\int_0^T\Theta\big(\sup_{x\in\IR^d}|a_t(x)|\big)+\Theta\big(\sup_{x\in\IR^d}|b_t(x)|\big)+\Theta\big(\sup_{x\in\IR^d}|c_t(x)|\big)\diff t<\infty,
\end{equation*}
for $\Theta:[0,\infty)\to[0,\infty)$ convex non-decreasing with $\lim_{x\to\infty} x^{-1}\Theta(x)=\infty$. Thus, in analogy to \textcite[Corollary~2.5]{trevisan2016well}, with $f^i=(x_i\chi_R)\circ e_t$ for $i=1,\dots,d$ and $f^{d+1}=[\log (y\chi_R)]\circ e_t$. Note that the function $f^{d+1}$ is always well-defined and real valued since $Y$ is assumed to be supported on $(0,e^T)$. Then, for the functional $\Psi:C([0,T];\IR)\to[0,\infty]$, we have for $i=1,\dots,d$, with $x^i_R=x_i\chi_R$:
\begin{align*}
\int\Psi(x^i_R\circ\gamma)\diff\boldsymbol{\eta}^{\epsilon}(\gamma)&\leq\int\theta(|x^i_R|)\diff\eta^{\epsilon}_0+\int_0^T\int\big[\Theta(|\mathcal{L}^{S,\epsilon}_tx^i_R|)\\
&\qquad+\Theta(b^{\epsilon}_t:\nabla x^i_R\otimes\nabla x^i_R)\big]\diff\eta^{\epsilon}_t\diff t,\\
\int\Psi(f^{d+1}\circ\gamma)\diff\boldsymbol{\eta}^{\epsilon}(\gamma)&\leq\int\theta(|f^{d+1}|)\diff\eta^{\epsilon}_0+\int_0^T\int \Theta(|c^{\epsilon}_t|)\diff\eta_t^{\epsilon}\diff t
\end{align*}
We now note that $\eta^{\epsilon}_t=\nu_t^{A,\epsilon}$ where $(\nu_t^{A,\epsilon})_{t\in[0,T]}$ solves $\mathrm{aFP}(a^{\epsilon},b^{\epsilon},c^{\epsilon})$, and further that by construction $\nu_0^{A,\epsilon}=\nu_0^{\epsilon}$. Thus, letting $R\to\infty$, noting that the $y$-marginal of $\eta^{\epsilon}_0$ is $\delta_1$, by the regular superposition principle (\textcite[Theorem~2.5]{trevisan2016well}), we have the uniform bounds:
\begin{align*}
\int\Psi(x^i_R\circ\gamma)\diff\boldsymbol{\eta}^{\epsilon}(\gamma)&\leq\int\theta(|x^i_R|)\diff\nu^{\epsilon}_0+\int_0^T\int\big[\Theta(|\mathcal{L}^{S,\epsilon}_tx^i_R|)\\
&\qquad+\Theta(b^{\epsilon}_t:\nabla x^i_R\otimes\nabla x^i_R)\big]\diff\nu^{A,\epsilon}_t\diff t,\\
\int\Psi(f^{d+1}\circ\gamma)\diff\boldsymbol{\eta}^{\epsilon}(\gamma)&\leq\theta(0)+\int_0^T\int \Theta(|c^{\epsilon}_t|)\diff\nu_t^{A,\epsilon}\diff t
\end{align*}
Thus, by \textcite[Lemma~A.1]{trevisan2016well}, we have the uniform bounds
\begin{align*}
\int\Psi(\gamma^i)\diff\boldsymbol{\eta}^{\epsilon}(\gamma)&\leq 1+\int_0^T\Theta\big(\sup_{x\in\IR^d}|a_t(x)|\big)+\Theta\big(\sup_{x\in\IR^d}|b_t(x)|\big)\diff t,\\
\int\Psi\big(\log(\gamma^{d+1})\big)\diff\boldsymbol{\eta}^{\epsilon}(\gamma)&\leq\theta (0) +\int_0^T\Theta\big(\sup_{x\in\IR^d}|c_t(x)|\big)\diff t.
\end{align*}
Since $\gamma\mapsto \Psi\big(\log(\gamma^{d+1})\big)+\sum_{i=1}^{d}\Psi(\gamma^i)$ is coercive in $C([0,T];\IR^{d+1})$, we have tightness of $\boldsymbol{\eta}^{\epsilon}$. The limit of the approximations obtaining a discounted superposition solution follows directly by the results of Section~\ref{sec: superposition convergence}.

\noindent
\subsubsection{Locally Bounded Coefficients:}

We now assume that $a,b,c$ satisfy
\begin{equation}\label{eq: local bdd coeffs}
\int_0^T\sup_{x\in B}\{|a_t(x)|+|b_t(x)|+|c_t(x)|\}\diff t<\infty,\qquad\text{for every bounded borel }B\subset\IR^d.
\end{equation}
Let $M\geq 1$ and $\chi_M$ be a cutoff function. Define $\pi_M:\IR^d\to\IR^d$ by $\pi_M(x)=x\chi_M(x)$. Thus, by \eqref{eq: local bdd coeffs}, we have $|\mathcal{L}^D(\pi^i_M)|\leq||\pi^i_M||_{C^2}\sup_{|x|\leq 2M}\{|a(x)|+|b(x)|+|x_rc(x)|\}$ and $|b:\nabla\pi^i_M\otimes\nabla\pi^j_M|\leq||\pi^i_M||_{C^1}\sup_{|x|\leq MR}|b(x)|$. Thus with $\nu^M_t\coloneqq\pi^M_{\#}\nu_t$, we have $(\nu_t^M)_{t\in[0,T]}$ solves $\mathrm{dFP}(\pi^M(a),\pi^M(b),\pi^M(c))$ with coefficients satisfying \eqref{eq: bounded coefficients} so we have $\boldsymbol{\eta}^M\in\cP\big(C([0,T];\IR^{d+1})\big)$ solving $\mathrm{aMP}\big(\pi^M(a),\pi^M(b),\pi^M(c)\big)$ as a discounted superposition solution to $\nu^M$. The argument of tightness is similar to that of the previous section with the modified bound to be
\begin{equation*}
\int_0^T\int\Theta_1(|a_t|)+\Theta_2(|b_t|)+\Theta_3(|c_t|)\diff\nu_t\diff t <\infty
\end{equation*}
As before, we obtain the following inequalities
\begin{align*}
\int\Psi(x^i_R\circ\gamma)\diff\boldsymbol{\eta}^{M}(\gamma)&\leq\int\theta(|x^i_R|)\diff\eta^{M}_0+\int_0^T\int\big[\Theta(|\mathcal{L}^{S,M}_tx^i_R|)\\
&\qquad+\Theta(b^{M}_t:\nabla x^i_R\otimes\nabla x^i_R)\big]\diff\eta^{M}_t\diff t,\\
\int\Psi(f^{d+1}\circ\gamma)\diff\boldsymbol{\eta}^{M}(\gamma)&\leq\int\theta(|f^{d+1}|)\diff\eta^{M}_0+\int_0^T\int \Theta(|c^{M}_t|)\diff\eta_t^{M}\diff t.
\end{align*}
By a similar reasoning to the smooth and bounded coefficients case, using Jensen's inequality, we obtain the following uniform bounds
\begin{align*}
\int\Psi(\gamma^i)\diff\boldsymbol{\eta}^M(\gamma)&\leq 1+\int_0^T\int\Theta_1(|a_t^i|)+\Theta_2(|b_t^{i,i}|)\diff\nu_t^A\diff t\\
\int\Psi\big(\log(\gamma^{d+1})\big)\diff\boldsymbol{\eta}^M(\gamma)&\leq\theta(0)+\int_0^T\int\Theta_3(|c_t|)\diff\nu_t^A\diff t.
\end{align*}
The limit of approximations obtaining a discounted superposition solution follows directly by Section~\ref{sec: superposition convergence}.

\noindent
\subsubsection{General Case:}

The approximation is achieved via convolution, with the approximating $\nu^{\epsilon}$ having superposition solutions since the approximating coefficients satisfy \eqref{eq: local bdd coeffs}. The tightness follows an identical argument, with similar bounds to the bounded coefficient case obtained. The limit of the approximations obtaining a discounted superposition solution follows directly by Section~\ref{sec: superposition convergence}. This concludes the proof of Theorem~\ref{thm: discounted superposition}.
\addcontentsline{toc}{section}{References}
\printbibliography

@article{guo2022calibration,
  title={Calibration of local-stochastic volatility models by optimal transport},
  author={Guo, Ivan and Loeper, Gr{\'e}goire and Wang, Shiyi},
  journal={Mathematical Finance},
  volume={32},
  doi={10.1111/mafi.12335},
  number={1},
  pages={46--77},
  year={2022},
  publisher={Wiley Online Library}
}

@article{guo2022joint,
  title={Joint Modeling and Calibration of {SPX} and {VIX} by Optimal Transport},
  author={Guo, Ivan and Loeper, Gr{\'e}goire and Ob{\l}{\'o}j, Jan and Wang, Shiyi},
  journal={SIAM Journal on Financial Mathematics},
  volume={13},
  number={1},
  pages={1--31},
  doi={10.1137/20M1375905},
  year={2022},
  publisher={SIAM}
}

@book{brigo2007interest,
  title={Interest rate models-theory and practice: with smile, inflation and credit},
  author={Brigo, Damiano and Mercurio, Fabio},
  year={2007},
  publisher={Springer Science \& Business Media}
}

@article{gyongy1986mimicking,
  title={Mimicking the one-dimensional marginal distributions of processes having an It{\^o} differential},
  author={Gy{\"o}ngy, Istv{\'a}n},
  journal={Probability Theory and Related Fields},
  volume={71},
  number={4},
  pages={501--516},
  doi={10.1007/BF00699039},
  year={1986},
  publisher={Springer}
}

@article{brunick2013mimicking,
  title={Mimicking an It{\^o} process by a solution of a stochastic differential equation},
  author={Brunick, Gerard and Shreve, Steven},
  journal={Annals of Applied Probability},
  volume={23},
  number={4},
  doi={10.1214/12-AAP881},
  pages={1584--1628},
  year={2013},
  publisher={Institute of Mathematical Statistics}
}

@article{trevisan2016well,
  title={Well-posedness of multidimensional diffusion processes with weakly differentiable coefficients},
  author={Trevisan, Dario},
  journal={Electronic Journal of Probability},
  volume={21},
  doi={10.1214/16-EJP4453},
  year={2016},
  publisher={The Institute of Mathematical Statistics and the Bernoulli Society}
}

@article{benamou2000computational,
  title={A computational fluid mechanics solution to the Monge-Kantorovich mass transfer problem},
  author={Benamou, Jean-David and Brenier, Yann},
  journal={Numerische Mathematik},
  volume={84},
  number={3},
  doi={10.1007/s002119900117},
  pages={375--393},
  year={2000},
  publisher={Springer}
}

@article{huesmann2019benamou,
  title={A Benamou--Brenier formulation of martingale optimal transport},
  author={Huesmann, Martin and Trevisan, Dario},
  journal={Bernoulli},
  volume={25},
  number={4A},
  doi={10.3150/18-BEJ1069},
  pages={2729--2757},
  year={2019},
  publisher={Bernoulli Society for Mathematical Statistics and Probability}
}

@article{hull1990pricing,
  title={Pricing interest-rate-derivative securities},
  author={Hull, John and White, Alan},
  journal={The Review of Financial Studies},
  volume={3},
  number={4},
  pages={573--592},
  doi={10.1093/rfs/3.4.573},
  year={1990},
  publisher={Oxford University Press}
}

@article{hull1994branching,
  title={Branching out},
  author={Hull, John and White, Alan},
  journal={Risk},
  volume={7},
  number={7},
  pages={34--37},
  year={1994}
}

@article{hull1995note,
  title={A note on the models of Hull and White for pricing options on the term structure: Response},
  author={Hull, John and White, Alan},
  journal={The Journal of Fixed Income},
  volume={5},
  number={2},
  pages={97--102},
  doi={10.3905/jfi.1995.408139},
  year={1995},
  publisher={Institutional Investor Journals Umbrella}
}

@article{dupire1994pricing,
  title={Pricing with a smile},
  author={Dupire, Bruno},
  journal={Risk},
  volume={7},
  number={1},
  pages={18--20},
  year={1994}
}

@incollection{guo2019local,
  title={Local volatility calibration by optimal transport},
  author={Guo, Ivan and Loeper, Gr{\'e}goire and Wang, Shiyi},
  booktitle={2017 MATRIX Annals},
  pages={51--64},
  year={2019},
  doi={10.1007/978-3-030-04161-8_5},
  publisher={Springer}
}

@book{villani2003topics,
  title={Topics in optimal transportation},
  author={Villani, C{\'e}dric},
  number={58},
  year={2003},
  publisher={American Mathematical Soc.}
}

@book {rockafellar1970convex,
    AUTHOR = {Rockafellar, R. Tyrrell},
     TITLE = {Convex analysis},
    SERIES = {Princeton Mathematical Series, No. 28},
 PUBLISHER = {Princeton University Press, Princeton, N.J.},
      YEAR = {1970}
}

@article{ambrosio2014well,
  title={Well-posedness of Lagrangian flows and continuity equations in metric measure spaces},
  author={Ambrosio, Luigi and Trevisan, Dario},
  journal={Analysis \& PDE},
  volume={7},
  number={5},
  pages={1179--1234},
  doi={10.2140/apde.2014.7.1179},
  year={2014},
  publisher={Mathematical Sciences Publishers}
}

@article{brenier1999minimal,
  title={Minimal geodesics on groups of volume-preserving maps and generalized solutions of the Euler equations},
  author={Brenier, Yann},
  journal={Communications on Pure and Applied Mathematics: A Journal Issued by the Courant Institute of Mathematical Sciences},
  volume={52},
  number={4},
  doi={10.1002/(SICI)1097-0312(199904)52:4%3C411::AID-CPA1%3E3.0.CO;2-3},
  pages={411--452},
  year={1999},
  publisher={Wiley Online Library}
}

@article{brenier1997homogenized,
  title={A homogenized model for vortex sheets},
  author={Brenier, Yann},
  journal={Archive for Rational Mechanics and Analysis},
  volume={138},
  number={4},
  pages={319--353},
  doi={10.1007/s002050050044},
  year={1997},
  publisher={Springer}
}

@article{tan2013optimal,
  title={Optimal transportation under controlled stochastic dynamics},
  author={Tan, Xiaolu and Touzi, Nizar},
  journal={The Annals of Probability},
  volume={41},
  number={5},
  doi={10.2307/42919802},
  pages={3201--3240},
  year={2013},
  publisher={Institute of Mathematical Statistics}
}

@inproceedings{krylov1984once,
  title={Once more about the connection between elliptic operators and It{\^o}’s stochastic equations},
  author={Krylov, Nikolai},
  booktitle={Statistics and Control of Stochastic Processes, Steklov Seminar},
  pages={214--229},
  year={1984}
}

@article{dupire1996unified,
  title={A unified theory of volatility},
  author={Dupire, Bruno},
  journal={Derivatives Pricing: The Classic Collection},
  pages={185--196},
  year={1996}
}

@article{loeper2006reconstruction,
  title={The reconstruction problem for the Euler-Poisson system in cosmology},
  author={Loeper, Gr{\'e}goire},
  journal={Archive for Rational Mechanics and Analysis},
  volume={179},
  number={2},
  doi={10.1007/s00205-005-0384-3},
  pages={153--216},
  year={2006},
  publisher={Springer}
}

@article{henrard2007irony,
  title={The Irony in the Derivatives Discounting},
  author={Henrard, Marc},
  journal={Wilmott Magazine},
  pages={92--98},
  year={2007}
}

@article{henrard2010irony,
  title={The irony in derivatives discounting part II: The crisis},
  author={Henrard, Marc},
  journal={Wilmott Journal},
  volume={2},
  number={6},
  doi={10.1002/wilj.39},
  pages={301--316},
  year={2010},
  publisher={Wiley Online Library}
}

@article{han2018solving,
  title={Solving high-dimensional partial differential equations using deep learning},
  author={Han, Jiequn and Jentzen, Arnulf and Weinan, E},
  journal={Proceedings of the National Academy of Sciences},
  volume={115},
  number={34},
  doi={10.1073/pnas.1718942115},
  pages={8505--8510},
  year={2018},
  publisher={National Acad Sciences}
}

@article{miltersen1997closed,
  title={Closed form solutions for term structure derivatives with log-normal interest rates},
  author={Miltersen, Kristian and Sandmann, Klaus and Sondermann, Dieter},
  journal={The Journal of Finance},
  volume={52},
  number={1},
  pages={409--430},
  doi={10.1111/j.1540-6261.1997.tb03823.x},
  year={1997},
  publisher={Wiley Online Library}
}

@article{brace1997market,
  title={The market model of interest rate dynamics},
  author={Brace, Alan and G{\k{a}}tarek, Dariusz and Musiela, Marek},
  journal={Mathematical Finance},
  volume={7},
  number={2},
  doi={10.1111/1467-9965.00028},
  pages={127--155},
  year={1997},
  publisher={Wiley Online Library}
}

@article{hagan2002managing,
  title={Managing smile risk},
  author={Hagan, Patrick and Kumar, Deep and Lesniewski, Andrew and Woodward, Diana},
  journal={The Best of Wilmott},
  volume={1},
  pages={249--296},
  year={2002}
}

@article{jamshidian1997libor,
  title={{LIBOR} and swap market models and measures},
  author={Jamshidian, Farshid},
  journal={Finance and Stochastics},
  volume={1},
  number={4},
  doi={10.1007/s007800050026},
  pages={293--330},
  year={1997},
  publisher={Springer}
}

@article{cox1996constant,
  title={The constant elasticity of variance option pricing model},
  author={Cox, John},
  journal={Journal of Portfolio Management},
  pages={15--17},
  doi={10.3905/jpm.1996.015},
  year={1996},
  publisher={Pageant Media}
}

@article{ren2007calibrating,
  title={Calibrating and pricing with embedded local volatility models},
  author={Ren, Yong and Madan, Dilip and Qian, Michael},
  journal={Risk Magazine},
  volume={20},
  number={9},
  pages={138},
  year={2007},
  publisher={RISK MAGAZINE LIMITED}
}

@article{engelmann2021calibration,
  title={Calibration of the Heston stochastic local volatility model: A finite volume scheme},
  author={Engelmann, Bernd and Koster, Frank and Oeltz, Daniel},
  journal={International Journal of Financial Engineering},
  volume={8},
  number={01},
  doi={10.1142/S2424786320500486},
  pages={2050048},
  year={2021},
  publisher={World Scientific}
}

@article{in2010adi,
  title={ADI finite difference schemes for option pricing in the Heston model with correlation},
  author={In't Hout, KJ and Foulon, S},
  journal={International Journal of Numerical Analysis \& Modeling},
  volume={7},
  number={2},
  pages={303--320},
  year={2010}
}

@article{ma2017unconditionally,
  title={An unconditionally monotone numerical scheme for the two-factor uncertain volatility model},
  author={Ma, K and Forsyth, PA},
  journal={IMA Journal of Numerical Analysis},
  volume={37},
  number={2},
  pages={905--944},
  doi={10.1093/imanum/drw025},
  year={2017},
  publisher={Oxford University Press}
}

@article{liu1989limited,
  title={On the limited memory {BFGS} method for large scale optimization},
  author={Liu, Dong and Nocedal, Jorge},
  journal={Mathematical Programming},
  volume={45},
  number={1},
  doi={10.1007/BF01589116},
  pages={503--528},
  year={1989},
  publisher={Springer}
}

@article{figalli2008existence,
  title={Existence and uniqueness of martingale solutions for SDEs with rough or degenerate coefficients},
  author={Figalli, Alessio},
  journal={Journal of Functional Analysis},
  volume={254},
  number={1},
  pages={109--153},
  doi={10.1016/j.jfa.2007.09.020},
  year={2008},
  publisher={Elsevier}
}

@article{guo2021path,
  title={Path dependent optimal transport and model calibration on exotic derivatives},
  author={Guo, Ivan and Loeper, Gr{\'e}goire},
  journal={The Annals of Applied Probability},
  volume={31},
  number={3},
  pages={1232--1263},
  doi={10.1214/20-AAP1617},
  year={2021},
  publisher={Institute of Mathematical Statistics}
}

@article{avellaneda1997calibrating,
  title={Calibrating volatility surfaces via relative-entropy minimization},
  author={Avellaneda, Marco and Friedman, Craig and Holmes, Richard and Samperi, Dominick},
  journal={Applied Mathematical Finance},
  volume={4},
  number={1},
  pages={37--64},
  doi={10.1080/135048697334827},
  year={1997},
  publisher={Taylor \& Francis}
}

@article{guo2022optimal,
  title={Optimal transport for model calibration},
  author={Guo, Ivan and Loeper, Gr{\'e}goire and Ob{\l}{\'o}j, Jan and Wang, Shiyi},
  journal={Risk Magazine},
  year={2022},
  publisher={Risk. net}
}

@article{barles1991convergence,
  title={Convergence of approximation schemes for fully nonlinear second order equations},
  author={Barles, Guy and Souganidis, Panagiotis},
  journal={Asymptotic Analysis},
  volume={4},
  number={3},
  pages={271--283},
  doi={10.3233/ASY-1991-4305},
  year={1991},
  publisher={IOS Press}
}

@article{lions1983optimal,
  title={Optimal control of diffusion processes and Hamilton--Jacobi--Bellman equations part 2: viscosity solutions and uniqueness},
  author={Lions, Pierre-Louis},
  journal={Communications in Partial Differential Equations},
  volume={8},
  number={11},
  pages={1229--1276},
  doi={10.1080/03605308308820301},
  year={1983},
  publisher={Taylor \& Francis}
}

@article{weinan2021algorithms,
  title={Algorithms for solving high dimensional {PDE}s: from nonlinear {M}onte {C}arlo to machine learning},
  author={Weinan, E and Han, Jiequn and Jentzen, Arnulf},
  journal={Nonlinearity},
  volume={35},
  number={1},
  pages={278--310},
  year={2021},
  doi={10.1088/1361-6544/ac337f},
  publisher={IOP Publishing}
}

@article{schmidt2005minfunc,
  title={minFunc: unconstrained differentiable multivariate optimization in Matlab},
  author={Schmidt, Mark},
  url={http://www.cs.ubc.ca/~schmidtm/Software/minFunc.html},
  year={2005}
}

@book{box1987empirical,
  title={Empirical model-building and response surfaces.},
  author={Box, George and Draper, Norman},
  year={1987},
  publisher={John Wiley \& Sons}
}

@book{nesterov2018lectures,
  title={Lectures on convex optimization},
  author={Nesterov, Yurii},
  volume={137},
  year={2018},
  publisher={Springer}
}

@article{crandall1992user,
  title={User’s guide to viscosity solutions of second order partial differential equations},
  author={Crandall, Michael and Ishii, Hitoshi and Lions, Pierre-Louis},
  journal={Bulletin of the American Mathematical Society},
  volume={27},
  number={1},
  pages={1--67},
  doi={10.1090/S0273-0979-1992-00266-5},
  year={1992}
}

@article{bain2021calibration,
  title={Calibration of local-stochastic and path-dependent volatility models to vanilla and no-touch options},
  author={Bain, Alan and Mariapragassam, Matthieu and Reisinger, Christoph},
  journal={Journal of Computational Finance},
  volume={24},
  doi={10.21314/JCF.2020.400},
  number={4},
  year={2021}
}

@article{bouchard2017hedging,
  title={Hedging of covered options with linear market impact and gamma constraint},
  author={Bouchard, Bruno and Loeper, Gr{\'e}goire and Zou, Yiyi},
  journal={SIAM Journal on Control and Optimization},
  volume={55},
  doi={10.1137/15M1054109},
  number={5},
  pages={3319--3348},
  year={2017},
  publisher={SIAM}
}

@article{guyon2020joint,
  title={The joint S\&P 500/VIX smile calibration puzzle solved},
  author={Guyon, Julien},
  journal={Risk, April},
  year={2020}
}

@article{guoNingloeper1,
  title={Portfolio optimization with a prescribed terminal wealth distribution},
  author={Guo, Ivan and Langren{\'e}, Nicolas and Loeper, Gr{\'e}goire and Ning, Wei},
  journal={Quantitative Finance},
  volume={22},
  number={2},
  pages={333--347},
  doi={10.1080/14697688.2021.1967432},
  year={2022},
  publisher={Routledge}
}

@book{bergomi2015stochastic,
  title={Stochastic volatility modeling},
  author={Bergomi, Lorenzo},
  year={2015},
  publisher={CRC press}
}

@misc{joseph2023joint,
      title={Joint Calibration of Local Volatility Models with Stochastic Interest Rates using Semimartingale Optimal Transport}, 
      author={Joseph, Benjamin and Loeper, Gr{\'e}goire and Ob{\l}{\'o}j, Jan},
      year={2023},
      eprint={2308.14473},
      archivePrefix={arXiv},
}

@misc{mBBRd1,
	author = {Backhoff-Veraguas, Julio and Beiglb{\"o}ck, Mathias and Schachermayer, Walter and Tschiderer, Bertram},
	title = {{The Structure of martingale {B}enamou-{B}renier in $\mathbb{R}^d$}},
	archivePrefix={arXiv},
	eprint={2306.11019},
	year = {2023}}

@article{wyns2017finite,
  title={A finite volume--alternating direction implicit approach for the calibration of stochastic local volatility models},
  author={Wyns, Maarten and Du Toit, Jacques},
  journal={International Journal of Computer Mathematics},
  volume={94},
  number={11},
  pages={2239--2267},
  doi={10.1080/00207160.2017.1297805},
  year={2017},
  publisher={Taylor \& Francis}
}

@article{ogetbil2022calibrating,
  title={Calibrating local volatility models with stochastic drift and diffusion},
  author={{\"O}getbil, Orcan and Ganesan, Narayan and Hientzsch, Bernhard},
  journal={International Journal of Theoretical and Applied Finance},
  volume={25},
  number={02},
  doi={10.1142/S021902492250011X},
  pages={2250011},
  year={2022},
  publisher={World Scientific}
}

@article{hok2019calibration,
  title={Calibration of local volatility model with stochastic interest rates by efficient numerical PDE methods},
  author={Hok, Julien and Tan, Shih-Hau},
  journal={Decisions in Economics and Finance},
  volume={42},
  number={2},
  doi={10.1007/s10203-019-00232-3},
  pages={609--637},
  year={2019},
  publisher={Springer}
}

@article{deelstra2013local,
  title={Local volatility pricing models for long-dated FX derivatives},
  author={Deelstra, Griselda and Ray{\'e}e, Gr{\'e}gory},
  journal={Applied Mathematical Finance},
  volume={20},
  number={4},
  doi={10.1080/1350486X.2012.723516},
  pages={380--402},
  year={2013},
  publisher={Taylor \& Francis}
}

@article{guyon2012being,
  title={Being particular about calibration},
  author={Guyon, Julien and Henry-Labord{\`e}re, Pierre},
  journal={Risk},
  volume={25},
  number={1},
  pages={88},
  year={2012},
  publisher={Incisive Media Limited}
}

@article{henry2009calibration,
  title={Calibration of local stochastic volatility models to market smiles: A Monte-Carlo approach},
  author={Henry-Labord{\`e}re, Pierre},
  journal={Risk Magazine, September},
  year={2009}
}

@misc{guyon2011smile,
  title={The smile calibration problem solved},
  author={Guyon, Julien and Henry-Labord{\`e}re, Pierre},
  archivePrefix={SSRN},
  doi={10.2139/ssrn.1885032},
  year={2011}
}

@article{cozma2019calibration,
  title={Calibration of a hybrid local-stochastic volatility stochastic rates model with a control variate particle method},
  author={Cozma, Andrei and Mariapragassam, Matthieu and Reisinger, Christoph},
  journal={SIAM Journal on Financial Mathematics},
  volume={10},
  number={1},
  doi={10.1137/17M1114570},
  pages={181--213},
  year={2019},
  publisher={SIAM}
}

@misc{atlan2006localizing,
  title={Localizing volatilities},
  author={Atlan, Marc},
  archivePrefix={arXiv},
  eprint={math/0604316},
  year={2006}
}

@book{stroock1979multidimensional,
  title={Multidimensional diffusion processes},
  author={Stroock, Daniel and Varadhan, Srinivasa},
  volume={233},
  year={1979},
  publisher={Springer Science \& Business Media}
}

@book{karatzas2014brownian,
  title={Brownian motion and stochastic calculus},
  author={Karatzas, Ioannis and Shreve, Steven},
  volume={113},
  year={2014},
  publisher={springer}
}

@misc{joseph2023measure,
  title={The Measure Preserving Martingale Sinkhorn Algorithm},
  author={Joseph, Benjamin and Loeper, Gr{\'e}goire and Ob{\l}{\'o}j, Jan},
  eprint={2310.13797},
  archivePrefix={arXiv},
  year={2023}
}

@book{ambrosio2005gradient,
  title={Gradient flows: in metric spaces and in the space of probability measures},
  author={Ambrosio, Luigi and Gigli, Nicola and Savar{\'e}, Giuseppe},
  year={2005},
  publisher={Springer Science \& Business Media}
}

@article{chizat2018unbalanced,
  title={Unbalanced optimal transport: Dynamic and Kantorovich formulations},
  author={Chizat, Lenaic and Peyr{\'e}, Gabriel and Schmitzer, Bernhard and Vialard, Fran{\c{c}}ois-Xavier},
  journal={Journal of Functional Analysis},
  volume={274},
  number={11},
  pages={3090--3123},
  doi={10.1016/j.jfa.2018.03.008},
  year={2018},
  publisher={Elsevier}
}

@article{sejourne2023unbalanced,
  title={Unbalanced optimal transport, from theory to numerics},
  author={S{\'e}journ{\'e}, Thibault and Peyr{\'e}, Gabriel and Vialard, Fran{\c{c}}ois-Xavier},
  journal={Handbook of Numerical Analysis},
  volume={24},
  pages={407--471},
  doi={10.1016/bs.hna.2022.11.003},
  year={2023},
  publisher={Elsevier}
}

@book{jacod2013limit,
  title={Limit theorems for stochastic processes},
  author={Jacod, Jean and Shiryaev, Albert},
  volume={288},
  year={2013},
  publisher={Springer Science \& Business Media}
}

@article{ambrosio2009flows,
  title={On flows associated to Sobolev vector fields in Wiener spaces: an approach {\`a} la DiPerna--Lions},
  author={Ambrosio, Luigi and Figalli, Alessio},
  journal={Journal of Functional Analysis},
  volume={256},
  number={1},
  doi={10.1016/j.jfa.2008.05.007},
  pages={179--214},
  year={2009},
  publisher={Elsevier}
}

@article{beiglbock2019dual,
author = {Mathias Beiglb{\"o}ck and Tongseok Lim and Jan Obł{\'o}j},
title = {{Dual attainment for the martingale transport problem}},
volume = {25},
journal = {Bernoulli},
number = {3},
publisher = {Bernoulli Society for Mathematical Statistics and Probability},
pages = {1640 -- 1658},
year = {2019},
doi = {10.3150/17-BEJ1015},
URL = {https://doi.org/10.3150/17-BEJ1015}
}
\end{document}